%% file: main.tex
\documentclass[11pt]{article}
\input{preamble}

\newcommand{\deficit}{\Delta}

\newcommand{\MR}{\textsc{Metric-Repair}}
\newcommand{\IOMR}{\textsc{IO-Metric-Repair}}

\newcommand{\IS}{\textsc{IndSet}}
\newcommand{\pMR}{\textsc{Planar-Metric-Repair}}
\newcommand{\gridMR}{\textsc{Grid-Metric-Repair}}
\newcommand{\pOrientation}{\textsc{Planar-3Deg-Nonsink-Orientation}}
\newcommand{\Partition}{\textsc{Partition}}

\newcommand{\Dnew}{D_{\mathrm{new}}}
\newcommand{\Dtemp}{D_{\mathrm{temp}}}

\newcommand{\Lnew}{\Lambda_{\mathrm{new}}}

\DeclareMathOperator{\cost}{cost}
\DeclareMathOperator{\DP}{\mathbf{DP}}
\DeclareMathOperator{\OPT}{\mathbf{OPT}}
\newcommand{\EdgeTable}{\mathrm{EdgeTable}}

\newcommand{\psbk}[1]{\left[ #1 \right]_+}
\author{Asaf Etgar \and Anna C.~Gilbert \and Jamie Tucker-Foltz}
\date{\today}

\begin{document}

\begin{titlepage}
\thispagestyle{empty}
\title{Structural Tractability Frontiers for Metric Repair}


\maketitle

\begin{abstract}
    Given a graph $G$ labeled with positive distances on each edge, what is the fewest number of edge distances that must be modified for $G$ to become a metric? It is known that this \emph{metric repair} problem is $\NPh$ on general graphs, with prior work focusing on approximations and fixed-parameter tractability with respect to properties of the input distance function. In this paper, we ask what structural properties of the graph itself make metric repair tractable.
    
    On the positive side, we give pseudo-polynomial time algorithms for series-parallel graphs, and by generalization, graphs of bounded treewidth. An immediate consequence of this result is a new algorithm for the \emph{length-bounded multicut} problem, with a parameterized runtime bound in terms of the treewidth of a modestly augmented graph.
    
    Surprisingly, pseudo-polynomial time turns out to be the best one can hope for: We complement our algorithm with a proof that metric repair is weakly $\NPh$ even on graphs of pathwidth at most six. We also prove that planarity does not help either, as the problem remains strongly $\NPh$ even on grid graphs.
\end{abstract}
\end{titlepage}

\maketitle

\section{Introduction}
\input{introduction}

\section{Series-Parallel graphs}\label{sec:sp}
\input{series-parallel-updated}

\section{FPT on Graphs of Bounded Treewidth}\label{sec:tw}
\input{bdd_tw_fpt}

\section{Hardness on Graphs of Bounded Treewidth}
\input{bdd-tw-is-hard}

\section{Hardness on Planar Graphs}\label{sec:planar}
\input{planar-is-hard}

\section{Hardness on Grids}\label{sec:grid}
\input{grid-is-hard}

\input{conclusion}

\bibliographystyle{alpha}
\bibliography{bib}
\end{document}

%% file: preamble.tex
\PassOptionsToPackage{full}{textcomp}
\usepackage[margin=1in]{geometry}
\usepackage{subcaption}
\usepackage{MathematicA} 
\usepackage{newtxtext}
\usepackage{newtxmath}
\newtheorem{theorem}{Theorem}[section]
\newtheorem{lemma}[theorem]{Lemma}
\newtheorem{claim}[theorem]{Claim}

\newtheorem{corollary}[theorem]{Corollary}

\newtheorem{observation}[theorem]{Observation}

\newtheorem{problem}[theorem]{Problem}

\theoremstyle{definition}
\newtheorem{definition}[theorem]{Definition}

\let\emph\relax 
\DeclareTextFontCommand{\emph}{\itshape}

\usepackage{hyperref}
\usepackage{cleveref}
\crefname{observation}{Observation}{Observations}
\Crefname{observation}{Observation}{Observations}
\crefname{statement}{Statement}{Statements}
\Crefname{statement}{Statement}{Statements}
\usepackage{xargs} 
\usepackage[x11names]{xcolor}	
\usepackage[colorinlistoftodos,prependcaption,textsize=tiny,linecolor=DeepSkyBlue1,backgroundcolor=DeepSkyBlue1!25,bordercolor=DeepSkyBlue1]{todonotes}
\usepackage[shortlabels]{enumitem} 

\newcommandx{\Todo}[2][1=]{\todo[linecolor=DeepSkyBlue1,backgroundcolor=DeepSkyBlue1!25,bordercolor=DeepSkyBlue1,#1]{#2}}

\newcommand{\rr}{\mathbb R}

\newcommand{\tth}{^\text{th}}

\newcommand{\ipncm}[3]{\begin{figure}[H]\begin{center}\includegraphics[scale = {#1}]{#2.pdf}\caption{#3}\end{center}\end{figure}}

%% file: introduction.tex
Distances are a fundamental building block of algorithm design and data analysis.
Clustering, nearest-neighbor search, low-distortion embedding, and
visualization all take as input a matrix of pairwise distances, and many
assume, implicitly or explicitly, that the matrix is a
\emph{metric}. That is, that the distances satisfy the triangle inequality. Real
distance data rarely cooperate. Measurements are noisy, distances are estimated
from incomplete observations or are stitched together from heterogeneous sources,
and the resulting data routinely \emph{violate} the triangle inequality.
This is more than a cosmetic defect. Much of the data analysis toolbox leans on the
triangle inequality for its correctness: nearest-neighbor indices use it to
prune the search space and can silently return the wrong neighbor when it fails,
low-distortion embeddings are only meaningful relative to a metric, and the
approximation guarantees of metric clustering algorithms such as $k$-center and
$k$-median simply cease to hold.

Faced with such data, a natural question is how to perturb it as
little as possible so that it becomes a genuine metric, one that is suitable for input to standard algorithms. This is the
\emph{metric repair} problem, introduced independently by Gilbert and
Jain~\cite{gilbert2017sparse} and by Fan, Raichel, and
Van~Buskirk~\cite{fan2018metric} under the name \emph{metric violation
distance}: given distances that violate the triangle inequality, change as few
of them as possible to repair the metric. The structure of the problem was
subsequently clarified by Fan, Gilbert, Raichel, Sonthalia, and
Van~Buskirk~\cite{fan_et_al:LIPIcs.SWAT.2020.25}, whose characterization
underlies our work.

We study metric repair in its graph formulation. The input is a graph
$G=(V,E,w)$ with positive edge weights $w$ (or distances) that are \emph{metric} when
no edge is strictly longer than some other path between its endpoints;
equivalently (\Cref{thm:fan-characterization}), when the graph contains no
\emph{broken} cycle---a cycle with one edge heavier than all of its other edges
combined. The \MR{} problem seeks the smallest set of edges whose weights
can be changed so as to repair every broken cycle. There are also natural variants of this problem where repaired weights can only be decreased or repaired weights can only be increased. The
\emph{increase-only} variant \IOMR{} is particularly natural for some applications, modeling the common setting in which measured distances are
underestimates. We focus on this variant in our algorithms, but remark that both our algorithms and hardness results extend to the general problem as well.

The combinatorial core of the metric repair problem is revealed by a structural theorem of
Fan~et~al.\ (\Cref{thm:fan-characterization}): A set of edges $S$ admits a
metric repair supported on $S$ if and only if $S$ is a \emph{hitting set} for
the broken cycles, containing at least one edge from each broken cycle. Moreover, in the decrease-only variant, we must hit the heavy edges, and in the increase-only variant, we must hit the light edges. This
observation reveals multiple aspects of this problem at once. If weights may only be
decreased, the problem is solvable in polynomial time, as each set we must intersect has cardinality one. As soon as
increases are permitted, however, metric repair is a minimum hitting
set problem over the (possibly exponentially large) family of broken cycles. Perhaps unsurprisingly, both \MR{} and \IOMR{} are $\NP$-hard on general
graphs~\cite{fan_et_al:LIPIcs.SWAT.2020.25}.

Beyond data cleaning, metric repair has a purely combinatorial motivation: it
generalizes the \emph{length-bounded multicut} problem. In length-bounded multicut one is
given a graph, terminal pairs $(s_i,t_i)_{i=1}^{k}$, and a bound $L$, and must
delete as few edges as possible so that no $s_i$--$t_i$ path of length at most
$L$ survives~\cite{dvovrak2018parameterized}. Add to the graph an edge
$s_it_i$ of weight $L$ for each terminal pair. Such an edge lies on a broken
cycle precisely when a short $s_i$--$t_i$ path exists, so a minimum metric
repair that is allowed to modify only the original edges is exactly a minimum
length-bounded cut. Length-bounded cuts are $\NP$-hard and have been studied
extensively from the parameterized and approximation viewpoints; metric repair
sits squarely above them. Our results can thus be read as a fine-grained map of
where this added generality remains tractable.
\newpage
\begin{quote}
Since metric repair is intractable in general, it is natural to ask on which
graph classes it becomes easy: \emph{Where is the boundary of tractability?}
\end{quote}

\begin{table}[t]
\centering
\renewcommand{\arraystretch}{1.2}
\begin{tabular}{lll}
\hline
\textbf{Graph class} & \textbf{Complexity} & \textbf{Reference}\\
\hline
Series-parallel & $O(W^{4}\cdot |E|)$                   & \Cref{sec:sp}\\

Treewidth $\le r$         & $W^{O(r^{2})}\cdot(|V| + |E|)$ \,\,\,(FPT) & \Cref{sec:tw}\\ \hline
Pathwidth $\le 6$, unbounded $W$ & (weakly) $\NP$-hard           & \Cref{sec:bdd-tw-hard}\\
\hline
General graphs            & $\NP$-hard                       & \cite{fan_et_al:LIPIcs.SWAT.2020.25}\\
Planar graphs             & $\NP$-hard (bounded integer weights) & \Cref{sec:planar}\\
Grid graphs               & $\NP$-hard                       & \Cref{sec:grid}\\
\hline
\end{tabular}
\caption{Summary of results. Hardness results are stated for \MR{} (arbitrary
weight changes); the algorithmic results are for the increase-only variant
\IOMR{}, with integer weights in $[W]$.}
\label{tab:results}
\end{table}

\subsection{Our results}

Our results fall into two complementary areas. On the algorithmic side, increase-only metric repair becomes tractable once \emph{both} the graph structure and the numerical weights
are bounded. The algorithms we present extend to the general metric-repair problem. We only present the increase-only algorithms, and discuss the generalization after each algorithm. On the hardness side, we show that neither restriction alone makes the problem tractable:
bounding the structure while leaving the weights free, or bounding the weights while leaving the structure planar, each keeps the problem $\NP$-hard. The boundary of tractability sits somewhere between pathwidth 6 and treewidth 2, and series-parallel graphs with unbounded weights. \Cref{tab:results} summarizes the landscape. 

\paragraph{Algorithms.}
Our main algorithmic result is that bounded treewidth, together with bounded
integer weights, makes increase-only metric repair fixed-parameter tractable.
\begin{theorem}[\Cref{sec:tw}]
Increase-only metric repair on graphs of treewidth at most $r$, with integer
weights in $[W]$, can be solved in time $W^{O(r^{2})}\cdot(|V|+|E|)$ given a nice
tree decomposition of width $r$. In particular, the problem is
fixed-parameter tractable parameterized by $(r,W)$.
\end{theorem}
\noindent
This algorithm also solves length-bounded multicut, improving on~\cite{dvovrak2018parameterized}:
adding a weighted edge between each terminal pair $s_i,t_i$ reduces the problem
to \IOMR{}. Where~\cite{dvovrak2018parameterized} pay an additive $k$ (the number of terminals) in the exponent by design, our reduction pays only
$\tw(G')$ for the augmented graph $G'$. This is never worse than $\tw(G)+k$, and
collapses to $O(\tw(G))$ when the pairs are joined by short, low-congestion
paths, eliminating the dependence on $k$ entirely for local pairs. 

\begin{corollary}[\Cref{sec:tw}]
    An instance $(G=(V,E),L)$ of length-bounded multicut with $k$ terminal pairs can be solved in time $L^{O(r^2)}\cdot(|V|+|E|+k)$, where $r=\tw(G')$ for the augmented graph $G'=(V,E\cup\{s_it_i:i\in[k]\})$.
\end{corollary}
\noindent For series-parallel graphs (the graphs of treewidth two) we give a
cleaner and faster algorithm, whose exposition also serves as a warm-up for the
general treewidth result.
\begin{theorem}[\Cref{sec:sp}]
Increase-only metric repair on series-parallel graphs, with integer weights in
$[W]$, can be solved in time $O(W^{4}\cdot |E|)$.
\end{theorem}

\paragraph{Hardness.}
Our hardness results show that both ingredients above are essential and that
planarity is no substitute for either. First, bounded treewidth alone does not
suffice: the dependence on the maximum weight $W$ cannot be removed.
\begin{theorem}[\Cref{sec:bdd-tw-hard}]
\MR{} and \IOMR{} are both weakly $\NP$-hard, even on graphs of pathwidth at
most $6$.
\end{theorem}
\noindent Since pathwidth bounds treewidth, this rules out any algorithm running
in time $f(\tw)\cdot\poly(|V|)$ independent of $W$, and so the pseudo-polynomial
$W$-dependence of our algorithms is unavoidable unless $\P=\NP$. Second,
restricting to a rigid \emph{planar} structure offers no relief either, even
when the weights are bounded by a constant. The reduction is from independent
set on cubic planar graphs.
\begin{theorem}[\Cref{thmPlanarHardness}]
\pMR{} is strongly $\NP$-hard, even when all edge weights are integers bounded by a
constant.
\end{theorem}
\noindent We push this all the way to the most constrained planar class, the
subgraphs of the integer grid.
\begin{theorem}[\Cref{sec:grid}]
\gridMR{} is strongly $\NP$-hard.
\end{theorem}

\subsection{Related work}

Metric repair was introduced independently by Gilbert and
Jain~\cite{gilbert2017sparse}, as \emph{sparse metric repair}, and by Fan, Raichel,
and Van Buskirk~\cite{fan2018metric}, under the name \emph{metric violation
distance}. In both formulations the input is a full distance matrix
that fails the triangle inequality, and the goal is to change as few
entries as possible to restore a metric. Fan et al.~\cite{fan2018metric}
proved hardness and approximation results for the general and
increase-only variants, and isolated the decrease-only variant as
polynomial-time solvable. More broadly, metric repair is part of the
literature on restoring metric structure to noisy distance data,
including metric nearness under norm objectives~\cite{BrickellDhillonSraTropp2008}
and recent work on fitting metrics and ultrametrics with minimum
disagreements~\cite{CohenAddadFanLeeMesmay2025}. These works concern
complete distance matrices; our focus is the graph-structured version.

Graph metric repair was introduced by Fan, Gilbert, Raichel, Sonthalia,
and Van Buskirk~\cite{fan_et_al:LIPIcs.SWAT.2020.25}. In this setting only the edges of a
weighted graph are specified, and feasibility is governed by
\emph{broken cycles}. Their broken-cycle/hitting-set characterization
is the structural backbone of our work: ordinary metric repair is
equivalent to hitting every broken cycle, while increase-only repair is
equivalent to hitting every broken cycle on a light edge. They also
proved NP-hardness on general graphs and gave approximation and
fixed-parameter results. We take a complementary view: rather than
approximating on general graphs, we map the exact tractability boundary
as a function of graph structure.

Metric repair also generalizes \textsc{Length-Bounded Cut}, in which
one deletes as few edges as possible so that no short $s$--$t$ path
survives. 
Length-bounded disjoint paths are already NP-hard for fixed length
bounds~\cite{IPS82}, and length-bounded cuts are NP-hard and hard to
approximate under related restrictions~\cite{BaierErlebachHallKohlerKolmanPangracSchillingSkutella2010,FHNN16}.
Dvo\v{r}\'{a}k and Knop~\cite{dvovrak2018parameterized} generalize to the multi-commodity version, \textsc{Length-Bounded Multicut}, which is still a special case of \IOMR{}. They
give algorithms on bounded-treewidth graphs whose dependence on the
length bound cannot in general be removed. Our bounded-treewidth
results exhibit a similar phenomenon for metric repair: tractability
requires bounding not only the treewidth but also the numerical weights.

Our algorithms use dynamic programming over graph decompositions. For
bounded treewidth, we use \emph{nice} tree decompositions and the standard
constructive theory for obtaining them~\cite{bodlaender1996efficient}; for
treewidth two we use the classical series-parallel decomposition of
Valdes, Tarjan, and Lawler~\cite{SP-graphs}. Our hardness
results use complementary decomposition ideas: a path-like gadget chain
for bounded-pathwidth hardness, planar local gadgets based on
NP-hardness of independent set in planar bounded-degree graphs
\cite{GareyJohnsonStockmeyer1976}, and rectilinear planar layouts of
Rosenstiehl and Tarjan~\cite{rosenstiehl1986rectilinear} to transfer
hardness to grid graphs.

\subsection{Technical overview}

Our results have two complementary parts.  On the algorithmic side, we
use specific graph decompositions and dynamic programming to show that increase-only-metric-repair becomes tractable when both the
graph structure and the numerical weights are bounded. The algorithms for the general metric repair problem follow the same techniques. On the hardness
side, we show that neither condition alone is sufficient via three different reductions, using a mixture of combinatorial and graph-theoretic arguments to cut down the vast space of possible solutions to the instances we construct.

\paragraph{Algorithms via distance--demand profiles.}
Our starting point is the broken-cycle characterization of
Fan et al.~\cite{fan_et_al:LIPIcs.SWAT.2020.25}.  A repaired weighting is metric precisely
when no edge is heavier than an alternative path between its endpoints.
Thus, after a dynamic program has processed a subgraph $H$, the
unprocessed part of the graph can threaten the metricity of $H$ only through paths that
enter and leave $H$ through its boundary.  The state must therefore
remember two kinds of information: the distances already realized
between boundary vertices, and the lower bounds that future outside
paths must satisfy in order not to make an already processed edge
heavy.

For series-parallel graphs, the boundary consists of two terminals
$s,t$.  A state consists of a truncated $s$--$t$ distance and a demand
value: the minimum length that any future outside $s$--$t$ path must
have in order not to create a broken cycle with an edge already
processed.  Series composition adds terminal distances and projects
demands through the new terminals.  Parallel composition takes the
shorter of the two terminal distances and checks the accumulated
demand.  Since repaired weights are integers in $[W]$, there are
$O(W^2)$ possible two-terminal profiles, and the resulting dynamic
program runs in time
\[
        O(W^4\cdot |E|).
\]
The bounded-treewidth algorithm is the same idea with a larger
boundary and a different graph decomposition, a \emph{nice} tree decomposition.  For a bag $B_t$ of size at most $r+1$, the state consists of
two matrices indexed by pairs of vertices in $B_t$: a truncated
distance matrix and a demand matrix. In a \emph{nice} tree decomposition, introduce-vertex and forget-vertex
nodes only add or remove isolated boundary vertices.  Introduce-edge
nodes guess the repaired length of the new edge, update the boundary
shortest-path matrix, and add the demand imposed by the new edge.  Join
nodes merge two processed subgraphs by taking the shortest-path closure
of the two boundary distance matrices and projecting the two demand
matrices through the merged distances.  Whenever a transition can create
new cycles across the boundary, metricity is tested by the inequality
``distance is not less than demand.''  The number of distance--demand profiles
is $W^{O(r^2)}$, yielding time
\[
        W^{O(r^2)}\cdot (|V|+|E|).
\]
Thus increase-only metric repair is fixed-parameter tractable
parameterized by $(r,W)$.

This also gives a new algorithm for length-bounded multicut: encoding each
terminal pair as a heavy edge forming the augmented graph $G'=(V,E\cup\{s_it_i:i\in[k]\})$, reduces it to increase-only metric repair. Where \cite{dvovrak2018parameterized} force all terminals into every bag of the decomposition, paying the number of terminals $k$ in the width by design, our demand matrix folds the pairs together by projection and entrywise maximum, so the running time depends only on $\tw(G')$. This is never worse than their bound and drops the dependence on $k$ entirely when the terminal pairs are local; we give a sufficient condition on $G$ under which this holds, namely that the pairs be joined by short paths of low congestion.

\paragraph{Why the weight dependence is necessary.}
The dependence on $W$ is unavoidable.  We reduce from
\textsc{Partition} to show that both \textsc{Metric-Repair} and
\textsc{IO-Metric-Repair} are NP-hard even on graphs of pathwidth at
most $6$, when weights are encoded in binary.  The construction is a
path-like chain of constant-size gadgets.  Each gadget contains two
locally broken triangles, forcing two modifications, and the surrounding
constraints force one of two possible patterns.  These two patterns
encode the choice of which side of the partition receives the number
$x_i$.  Two global cycles, with heavy edges of weight $4n+1/2$, check
whether the two selected sums are both at most $1/2$.  Hence a repair
of size at most $2n$ exists if and only if the \textsc{Partition}
instance is feasible.  This weak hardness matches the algorithmic
result: bounded treewidth gives tractability only when the weights are
also bounded.

\paragraph{Planar hardness via local gadgets.}
We next show that planarity does not make metric repair easy, even when there are a \emph{constant} number of possible weights.  The
reduction is via an intermediate problem about edge orientation which is equivalent to independent set on planar
maximum-degree-three graphs.  Each vertex of the source graph is
replaced by a vertex gadget: a 5-cycle with one heavy edge of weight
$13$.  Each edge of the source graph is replaced by an edge gadget: a
chain of three triangles glued to the two incident vertex gadgets.

The edge gadget enforces an orientation choice.  If an edge is oriented
from $u$ to $v$, then the terminal edge on the $v$ side is repaired,
together with one shared edge inside the edge gadget.  Thus each edge
gadget contributes exactly two modifications.  A vertex whose incident
edges all point into it has all three terminal edges repaired and its
heavy edge satisfies the triangle inequality for free.  A nonsink
vertex must pay one additional modification in its vertex gadget.
Therefore a source graph with $m$ edges has an orientation with at most
$s$ nonsinks if and only if the constructed metric-repair instance has a
repair of size at most $2m+s$.  The construction is planar, has maximum
degree $6$, uses only the weights $1,4,13$, and proves hardness for both
ordinary and increase-only metric repair.

\paragraph{Grid hardness via rectilinear layouts.}

Finally, we strengthen planar hardness to grid graphs.  Starting from a
planar metric-repair instance, the basic approach is embed the graph as a grid minor. However, obstacles quickly emerge.

We use a rectilinear layout in which each
vertex is represented by a horizontal path and each edge by a vertical
path. The main challenge is that the grid contains many auxiliary cycles that
do not correspond to cycles of the source graph.  The construction
therefore reshapes the layout before assigning weights.

First, we dilate the layout to separate distinct edge paths.  Then each
vertex path is thickened into a large rectangular block.  The structural
part of the resulting grid consists of the edge paths and vertex blocks.
On each edge path, one distinguished segment carries the original edge
weight; all other structural segments have tiny weight $\varepsilon$.
Auxiliary segments receive maximum weight.  The induced-skeleton
argument shows that every broken cycle in the grid lies entirely in the
structural skeleton and is exactly a realization of a broken cycle in
the source graph.

The remaining issue is soundness: A small grid repair should not be
able to cut through a vertex block in a way that has no analogue in the
source graph.  The block construction prevents this. Our key lemma (\Cref{lem:grid-block-segment-hits-atmost-two}) uses planar duality to show that any such cut must involve too many edges to be a good solution.  This allows us to canonicalize any solution into one that corresponds to a solution in the source graph: auxiliary segments can be removed, multiple
segments on the same edge path can be replaced by the distinguished
segment, and vertex-block segments can be shifted onto edge paths.
After this, every grid repair induces a repair of no larger
size in the original graph, while every original repair lifts directly
to the grid.  This proves NP-hardness of \MR{} and \IOMR{} on grid
graphs.

\subsection{Organization}
We lay out our notation and central problems in~\Cref{sec:notation}.
 \Cref{sec:sp} presents the series-parallel
algorithm, which \Cref{sec:tw} generalizes to bounded treewidth. \Cref{sec:bdd-tw-hard} establishes weak NP-hardness on graphs of bounded pathwidth.
\Cref{sec:planar} proves hardness on planar graphs and \Cref{sec:grid}
strengthens it to grid graphs.  

\section{Notation and Problem Statement}
\label{sec:notation}
Unless stated otherwise, graphs $G = (V,E)$ are simple and undirected. A cycle in $G$ is a connected subgraph where every vertex has degree $2$. We think of a cycle $C$ as a sequence of edges $e_1,\ldots e_k$. If $G$ is associated with an edge-weight function $w:E\to \RR_{>0}$, for a cycle $C$ the \emph{deficit} of $C$ is $w(e) - w(C\setminus e)$ where $e$ is the edge of maximum weight in $C$. We say that $C$ is \emph{broken} if its deficit is positive. For a broken cycle $C$, the edge of maximal weight in $C$ is called \emph{heavy} and the other edges are called \emph{light}. Note that an edge can be heavy in one cycle and light in another; when we say that an edge is heavy or light, it will be clear from context which cycle we refer to. The set of all broken cycles is denoted $\Cc = \Cc(G)$. We say that a set $S\subset E$ is a \emph{hitting set} for $\Cc$ if $S$ intersects every broken cycle. We say that it is a \emph{light hitting set} if for every broken cycle $C$, the set $C\cap S$ contains some light edge of $C$. If there are no broken cycles under $w$, we say that $w$ is \emph{metric}, and otherwise \emph{non-metric}.

Given a positively weighted graph $G = (V,E,w)$ and a new weight function $\rho:E\to \RR_{>0}$, the set $S_\rho \defeq \cbk{e\in E\mid w(e) \ne \rho(e)}$ is called the \emph{support} of $\rho$ with respect to $w$. The \emph{Metric Repair} problem asks to find a metric weight function $\rho:E\to \RR_{>0}$ whose support $|S_\rho|$ is as small as possible. The \emph{Increase Only Metric Repair} (IOMR) problem adds the constraint $\rho(e)\ge w(e)$ for all $e\in E$. As it turns out, the combinatorial structure of the problem encapsulates the entirety of the solution: in~\cite{fan_et_al:LIPIcs.SWAT.2020.25} Fan et al.\ show the following theorem.

\begin{theorem}[\cite{fan_et_al:LIPIcs.SWAT.2020.25}, Theorem 6, Proposition 7]\label{thm:fan-characterization}
Let $G = (V,E,w)$ be a positively weighted graph and $S\subset E$. Then
\begin{enumerate}
    \item There exists a metric $\rho:E\to \RR_{>0}$ supported on $S$ if and only if $S$ is a hitting set for $\Cc$.
    \item There exists a metric $\rho:E\to \RR_{>0}$ such that $\rho\ge w$ and $\rho$ is supported on $S$ if and only if $S$ is a light hitting set for $\Cc$.
\end{enumerate}
Moreover, there is a polynomial time algorithm that, given a hitting set $S$, produces a suitable $\rho$.
\end{theorem}
Therefore, we define the metric repair problem as follows. 

\begin{problem}[\MR]
    Given a graph $G$ and a weight function $w:E\to \RR_{>0}$, is there a set of edges $S$ of size at most $k$ that intersects every broken cycle in $(G,w)$?
\end{problem}
\begin{problem}[\IOMR{}]
    Given a graph $G$ and a weight function $w:E\to \RR_{>0}$, is there a set of edges $S$ of size at most $k$ that intersects every broken cycle in $(G,w)$ at a light edge?
\end{problem}
We note that for each of these problems, the following three formulations are equivalent:
\begin{enumerate}
    \item \label{formulation:find x} finding a (increase-only) weight function $\rho$ with smallest support;
    \item \label{formulation: hitting set} finding a minimum (light) hitting set $S$ of $\Cc$;
    \item \label{formulation: LB cut} finding a minimum (light) set $S\subset E$ such that $G\setminus S$ has no broken cycles.
\end{enumerate}
While all are equivalent, different proofs use different formulations for simplicity of presentation. We transition between these formulations freely.

%% file: series-parallel-updated.tex
We start by presenting the main ideas of our dynamic program for series parallel graphs. Since all graphs we consider are $2$-connected, the series parallel graphs are precisely the graphs of treewidth at most $2$, which makes the dynamic program more digestible. 
\subsection{Set up}
We adapt the recursive definition of~\cite{SP-graphs} to the undirected setting. 

A graph $G$ is called \emph{Series-Parallel} (SP) with terminals $\{s,t\}$ if it is one of the following:
\begin{enumerate}
    \item A single edge $st$
    \item $G$ is the union of two edge disjoint graphs SP graphs $G_1,G_2$ in one of two ways:
    \begin{itemize}
        \item \textbf{Parallel: } $G_1,G_2$ share the same terminal set $\{s,t\}$, and $V(G_1)\cap V(G_2) = \{s,t\}$.
        \item \textbf{Series: } $G_1,G_2$ have terminal sets $\{s,l\},\{l,t\}$ respectively and $V(G_1)\cap V(G_2) = \{l\}$.
    \end{itemize}

\end{enumerate}
\begin{figure}[!ht]
    \centering
    \includegraphics[scale=.4]{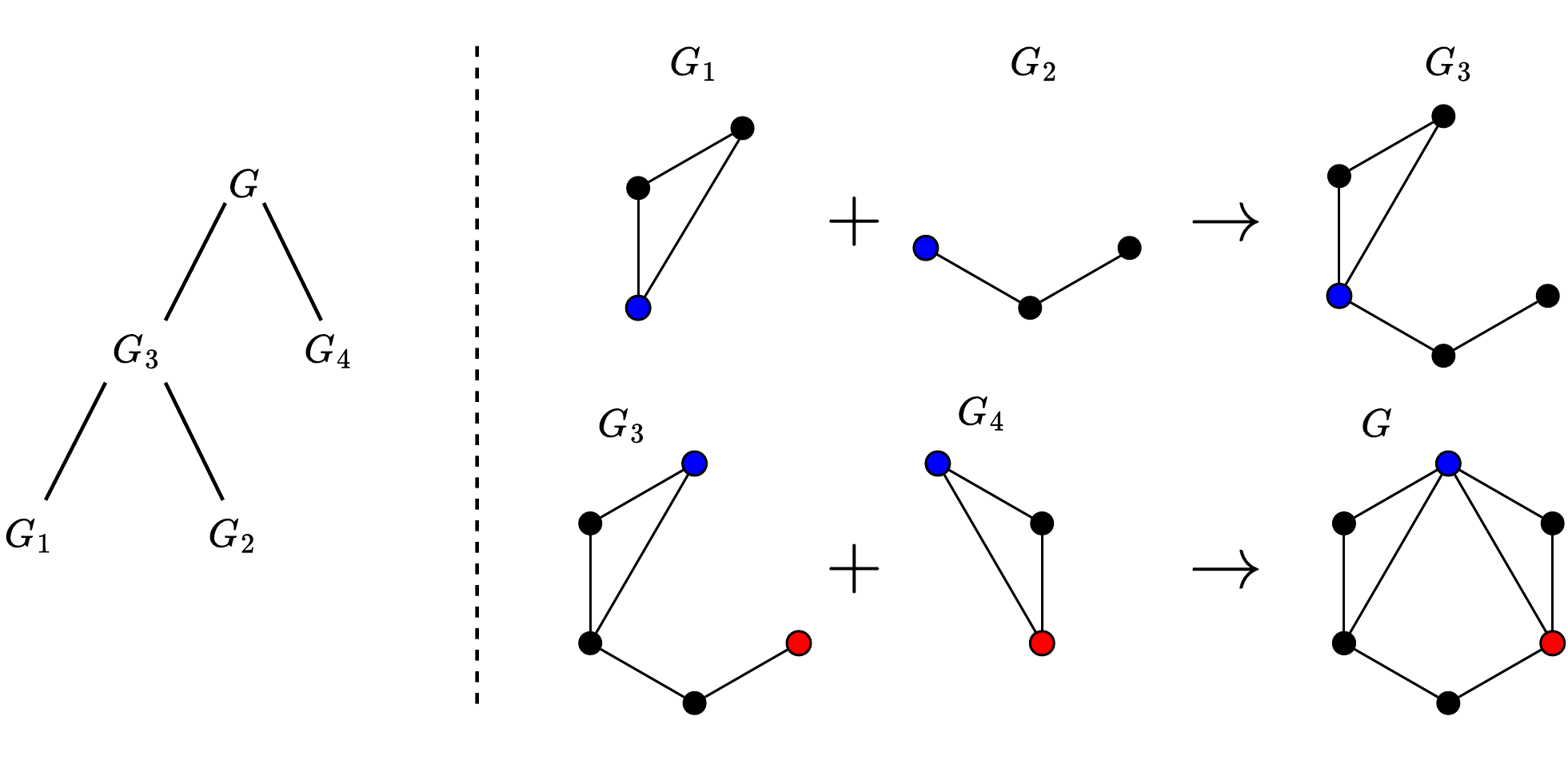}
    \caption{Series (top) and parallel (bottom) composition of Series-Parallel graphs. The matching terminals in each graph are colored, and the tree structure is on the left.}
    \label{fig:SPexample}
\end{figure}
This recursive structure allows one to construct a \emph{decomposition tree} of an SP graph: if $G$ is a single edge $st$, the tree is a single node representing $st$. If $G$ is a composition of $G_1,G_2$, the tree is rooted at a node $r$ representing $G$, with two children $f_1,f_2$, The subtree rooted at $f_i$ is the decomposition tree of $G_i$. Note that this is a binary tree with leaves corresponding to all edges of $G$. Constructing this tree takes linear time~\cite{SP-graphs}. Our dynamic program operates on this tree. 

To set notation, we are given a binary tree $T$ that is a SP decomposition tree of $G$. Its leaves are indexed by the edges of $G$, and each internal node is indexed by an SP subgraph $H$ of $G$, with decomposition tree rooted at that node. The children of $H$ are $H_1,H_2$. We say that a node is a Parallel (Series) node if $H$ is a Parallel (Series) composition of $H_1,H_2$.

A dynamic program state for an SP graph $H$ with terminals $\{s,t\}$ is comprised of two integers $(d,\lambda)$. $d$ is the \emph{distance} between $s,t$ in $H$, and $\lambda$ is a \emph{demand} value, encoding what the weight of a path of a future $s-t$ path needs to be in order to not cause edges already processed in $H$ to become heavy. States we process will correspond to increase only functions $\rho:E(H)\to [W]$ that are metric on $H$. However, the program does not maintain the functions $\rho$ explicitly. Each such function corresponds to possible increase values for the edges in $H$, corresponding to the leaves in the subtree rooted at $H$. The DP table for $H$ is denoted $A_H$. We begin by describing the DP tables for the leaves -- those are \emph{Single-edge tables}. We proceed by describing what DP states and tables for a general two terminal graph look like, namely \emph{Two-terminal profiles}. We then introduce a \emph{projection} operation, allowing us to design the DP update step for Series and Parallel nodes. Finally, we prove the correctness of these updates, and the correctness of the dynamic program.

\paragraph{Single-edge tables.}
For an edge $e=st$, define $\EdgeTable(e)$ as follows. For
each length
\[
\ell\in\{w(e),w(e)+1,\ldots,W\},
\]
We define the state
\[
(d,\lambda) \defeq (\ell,\ell).
\]
The cost of this state is
\[
\cost(\ell)=
\begin{cases}
0, & \ell=w(e),\\
1, & \ell>w(e).
\end{cases}
\]

Note that for each guess $\ell$, the shortest $s-t$ path in the graph with the unique edge $st$ is of length $\ell$, and any future path $P$ between $s$ and $t$ needs to have length at least $\ell$ so that $st$ is not heavy in $P\cup \{st\}$. 
\paragraph{Two-terminal profiles.}
Let $f$ be an internal node of $T$ representing the SP graph $H$ with terminals $S = \{s,t\}$.
For an increase-only repaired weight function
\[
\rho:E(H)\to [W],
\]
we define the \emph{two-terminal profile realized by $\rho$} to be the pair of integers $d_\rho,\lambda_\rho$ as follows.
\[
d_\rho=\tau\bigl(\dist_{\rho,H}(s,t)\bigr)
\]
where $\dist_{\rho,H}(s,t)$ is the shortest path distance between $s,t$ in $H$ under the weight function $\rho$, and
\[
\tau(r) \defeq \min\cbk{r,W+1}.
\]
Thus $d_\rho\in\{0,1,\ldots,W+1\}$.

Second, $\lambda_\rho$ is defined by
\[
\lambda_\rho
=
\psbk{
\max_{ab\in E(H)}
\left(
\rho(ab)-\dist_{\rho,H}(a,s)-\dist_{\rho,H}(b,t)
\right)
},
\]
where
\[
[r]_+ \defeq \min\{W,\max\{0,r\}\}.
\]
and the maximum ranges over both ordered pairs $(a,b)$ and
$(b,a)$ for each undirected edge $ab$. Hence
$\lambda_\rho\in\{0,1,\ldots,W\}$.

Intuitively, $\lambda_\rho(s,t)$ records how long a future outside route from
$s$ to $t$ must be in order not to make an already processed edge of
$H$ heavy. We say that a pair of integers $(d,\lambda)$ is \emph{realized} by $\rho$ if $d = d_\rho$ and $\lambda = \lambda_\rho$.
Let $\Xx(d,\lambda) $
be the set of all increase-only functions $\rho:E(H)\to[W]$ that are metric on
$H$ and realize $(d,\lambda)$.

The dynamic programming table for $f$ stores
\[
A(d,\lambda)
=
\displaystyle
\min_{\rho\in\Xx(d,\lambda)}
\cost(\rho),
\]
where 
\[
\cost(\rho) \defeq \abs{\cbk{e\in E(H) : \rho(e) > w(e)}}.
\]
We later show that we never keep states that are realized by non-metric $\rho$.
The dynamic program never stores the function $\rho$ explicitly. Nevertheless,
every finite table entry should be interpreted as being represented by at least one such function. 

This interpretation is maintained inductively by the transitions. For a
single-edge table, choosing a repaired length $\ell\in\{w(e),\ldots,W\}$
defines the realizing function on that one edge. For an internal face node, the
transition combines child states. By induction, each finite child state has at least one realizing function. More formally, for an SP graph $H$ that is the composition of two edge disjoint SP graphs $H_1$,$H_2$ and corresponding weight functions $\rho_i:E(H_i)\to [W]$, we denote by $\rho_1\cup \rho_2$ the weight function on $E(H)$ defined by 
\[ (\rho_1\cup \rho_2)(e)= \begin{cases} \rho_1(e), & e\in E(H_1),\\ \rho_2(e), & e\in E(H_2). \end{cases} \]
and note that 
\[
\cost(\rho_1\cup \rho_2) = \cost(\rho_1) + \cost(\rho_2).
\]


The two types of operations we preform, namely parallel and series composition, compute the two-terminal profile realized by this union function. Thus every
state produced by the dynamic program is again realized by some weight function $\rho$.

\paragraph{Projection of demands.}

When constructing an SP graph $H$ by composition of $H_1,H_2$, we will need to establish what a demand between the two terminals of $H$ is, given the demands between the terminals of $H_1,H_2$. Intuitively speaking, if an old demand says that any future route between old terminals
$p,q$ must have length at least $\lambda$, then after moving to new
terminals $s,t$, an outside $s$-to-$t$ route can threaten this old demand
only after paying the internal distances from $p$ to $s$ and from $q$ to $t$. If $\lambda$ is a demand between terminals $(s,t)$, we will sometime write $\lambda(s,t)$ to remind the reader that these are the terminals $\lambda$ refers to.

Let $H_i$ be an SP graph with terminal set $S_i=\{p_i,q_i\}$, and suppose $H_i$ is contained in a larger SP graph $H$ with
terminal set $S=\{s,t\}$. Suppose that $H$ is equipped with a weight function $x:E(H)\to [W]$ and let $\rho_i$ be its restriction to $H_i$. Let $D_H$ be the truncated distance matrix among vertices in $S_i\cup S$ under the weight function $\rho$. If $(d_i,\lambda_i)$ is a state for $H_i$ realized by $\rho_i$, define
its \emph{projected demand} on $S$ by
\[
\Pi_{H_i\to S}(s,t)
=
\left[
\max_{p,q\in S_i}
\left(
\lambda_i(p,q)-D_H(p,s)-D_H(q,t)
\right)
\right]_+
\]

\begin{lemma}[Projection of Demands]\label{lem:outerplanar-projection}
Let $H_i,H, S_i,S, x,\rho_i$ be as above. Then 
\[
\Pi_{H_i\to S}(s,t) = \psbk{\max_{ab=e\in E(H_i)}(\rho(ab) - \dist_{\rho,H}(a,s) - \dist_{\rho,H}(b,t)}
\]
where the maximum is taken over both orientations of $ab\in E(H_i)$.    
\end{lemma}
\begin{proof}
    Consider an edge $ab\in E(H_i)$. Since $S_i$ disconnects $H_i$ from $G$, it also disconnects $H_i$ from $H$. Thus, every path from $a$ to $s$ must pass through a terminal in $S_i$ to reach $s$. Hence
\[
\dist_{\rho,H}(a,s)
=
\min_{p\in S_i}
\left(
\dist_{\rho_i,H_i}(a,p)+D_H(p,s)
\right).
\]
Similarly,
\[
\dist_{\rho,H}(b,t)
=
\min_{q\in S_i}
\left(
\dist_{\rho_i,H_i}(b,q)+D_H(q,t)
\right).
\]
Therefore
\[
\begin{aligned}
&\rho(ab)-\dist_{\rho,H}(a,s)-\dist_{\rho,H}(b,t) \\
&=
\rho(ab)
-
\min_{p\in S_i}
\left(
\dist_{\rho_i,H_i}(a,p)+D_H(p,s)
\right)
-
\min_{q\in S_i}
\left(
\dist_{\rho_i,H_i}(b,q)+D_H(q,t)
\right) \\
&=
\max_{p,q\in S_i}
\left(
\rho(ab)
-\dist_{\rho_i,H_i}(a,p)-\dist_{\rho_i,H_i}(b,q)
-D_H(p,s)-D_H(q,t)
\right).
\end{aligned}
\]
Taking the maximum over all edges $ab\in E(H_i)$, we get
\[
\begin{aligned}
&\max_{\substack{ab\in E(H_i)}}
\left(
\rho(ab)-\dist_{\rho,H}(a,s)-\dist_{\rho,H}(b,t)
\right) \\
&=
\max_{p,q\in S_i}
\left(
\max_{ab\in E(H_i)}
\Bigl(
\rho(ab)-\dist_{\rho_i,H_i}(a,p)-\dist_{\rho_i,H_i}(b,q)
\Bigl)
-D_H(p,s)-D_H(q,t)
\right).
\end{aligned}
\]
By correctness of $\lambda_i$, the inner maximum is exactly the unprojected
demand represented by $\lambda_i(p,q)$. Since all edge lengths are at most
$W$, these demands are never larger than $W$, and negative values remain
irrelevant after subtracting the nonnegative distances $D_H(p,s),D_H(q,t)$. Thus
the positive clipping $[\cdot]_+$ commutes with this projection step, giving
\[
\left[
\max_{ab\in E(H_i)}
\left(
\rho(ab)-\dist_{\rho,H}(a,s)-\dist_{\rho,H}(b,t)
\right)
\right]_+
=
\left[
\max_{p,q\in S_i}
\left(
\lambda_i(p,q)-D_H(p,s)-D_H(q,t)
\right)
\right]_+ .
\]
\end{proof}

\begin{figure}[hb]
    \centering
    \includegraphics[scale = .5]{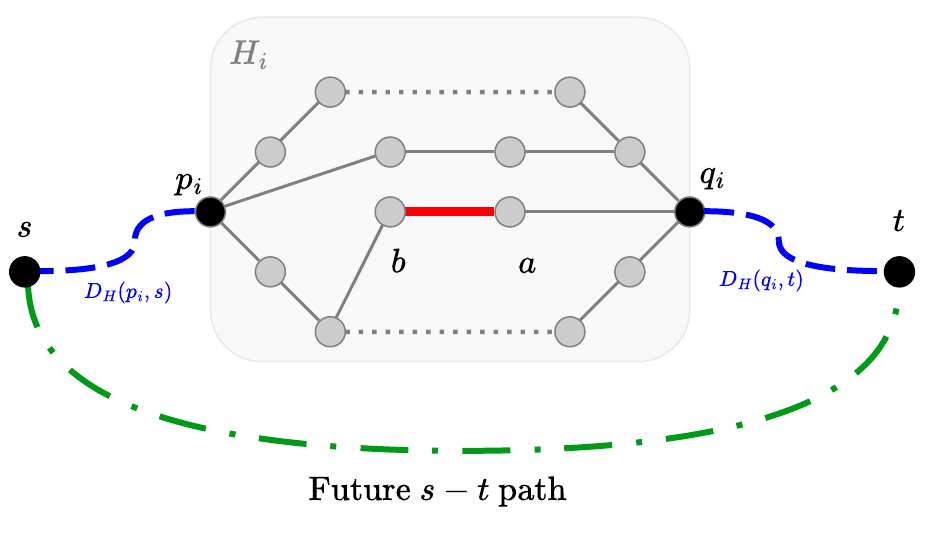}
    \caption{Illustration of demand projection. If an edge $ab$ sets a demand on an $p_i-q_i$ path, any path that gets from $a,b$ to $s,t$ needs to pay the cost of the $p_i-s$ and $q_i-t$ distance.}
\end{figure}
\begin{corollary}\label{cor:outerplanar-projection-union}
Suppose $H$ is a composition of two edge disjoint SP graphs $H_1,H_2$ and $S=\{s,t\}$ is the terminal set of $H$. Suppose $\rho$ is a metric repair function on $E(H)$. For $i\in \{1,2\}$, let $\rho_i$ be its restriction to $E(H_i)$ and let  $(d_i,\lambda_i)$ be a state realized by $\rho_i$. Let $\lambda$ be the $s-t$ demand realized by $\rho$. Then
\[
\lambda
=
\max\left\{
\Pi_{H_1\to S}(s,t),
\Pi_{H_2\to S}(s,t)
\right\}
\]
\end{corollary}
\begin{proof}
    By \Cref{lem:outerplanar-projection}, 
    \[
    \Pi_{H_i\to S}(s,t) = \psbk{\max_{ab=e\in E(H_i)}(\rho(ab) - \dist_{\rho,H}(a,s) - \dist_{\rho,H}(b,t)}.
    \]
    Since $E(H) = E(H_1)\cup E(H_2)$, 
    \[
    \max\left\{
    \Pi_{H_1\to S}(s,t),
    \Pi_{H_2\to S}(s,t)
    \right\} = \psbk{ \max_{ab\in E(H)} 
    \left( 
    \rho(ab)-\dist_{\rho,H}(a,s)-\dist_{\rho,H}(b,t) 
    \right)
    } = \lambda(s,t).
    \]
\end{proof}

We are now ready to describe the table updates for Parallel and Series nodes. 
\paragraph{Parallel composition.}

Let $H_1,H_2$ be two edge-disjoint SP graphs with the same terminals $s,t$ and DP tables $A_1,A_2$. Let $H$ be their parallel composition of $H_1,H_2$.
Consider states $(d_1,\lambda_1)$ and $(d_2,\lambda_2)$ for $H_1,H_2$. We construct a new state $(d,\lambda)$ for $H$ by defining
\[
d=\min\{d_1,d_2\}\qquad\text{and}\qquad \lambda
=
\max\left\{
\Pi_{H_1\to\{s,t\}}(s,t),
\Pi_{H_2\to\{s,t\}}(s,t)
\right\}
\]

We check
\[
d\ge \lambda
\]
If this check passes, we update
\[
A_H(d,\lambda)
=
\min\left\{
A_H(d,\lambda),
A_{H_1}(d_1,\lambda_1)+A_{H_2}(d_2,\lambda_2)
\right\}.
\]
Otherwise we discard the guess $(d,\lambda)$.

\paragraph{Series composition.}

Let $H_1,H_2$ be two edge-disjoint SP graphs with terminal sets $\{s,l\}$ and $\{l,t\}$ respectively. Let $H$ be their series composition with terminals $\{s,t\}$.
Consider states $(d_1,\lambda_1)$ and $(d_2,\lambda_2)$ for $H_1,H_2$. We construct a new state $(d,\lambda)$ for $H$ by defining
\begin{flalign*}
    d = \tau(d_1 + d_2) \qquad \text{and}\qquad
    \lambda
=
\max\left\{
\Pi_{H_1\to\{s,t\}}(s,t),
\Pi_{H_2\to\{s,t\}}(s,t)
\right\}.
\end{flalign*}
Update
\[
A_H(d,\lambda)
=
\min\left\{
A_H(d,\lambda),
A_{H_1}(d_1,\lambda_1)+A_{H_2}(d_2,\lambda_2)
\right\}.
\]

The subsequent lemmas establish the correctness of updates for both types of nodes. Most arguments hold for both compositions, but some require specific details for each type of composition. We specify when this is the case. 

\begin{lemma}[Correctness of Updates]\label{lem:SP-update}
    Suppose that for $i\in \{1,2\}$, $(d_i,\lambda_i)$ is a state of $H_i$ realized by $\rho_i$. Let $\rho = \rho_1\cup \rho_2$, and let $(d,\lambda)$ be their composition. Then $d = d_\rho$ and $\lambda = \lambda_\rho$.
\end{lemma}
\begin{proof}
    For both series and parallel composition, correctness of $\lambda$ follows from \Cref{cor:outerplanar-projection-union}.

    Suppose $H$ is the parallel composition of $H_1,H_2$.
    Since $H_1,H_2$ share the same terminal set $\{s,t\}$, any $s-t$ shortest path in $H$ remains entirely in either $H_1$ or $H_2$. Therefore
    \[
    d_\rho = \dist_\rho(s,t) = \min\cbk{\dist_{\rho_1,H_1}(s,t),\dist_{\rho_2,H_2}(s,t)} = \min\cbk{d_1,d_2}.
    \]

    Suppose $H$ is the series composition of $H_1,H_2$ with terminal sets $\{a,b\}$ and $\{b,c\}$ respectively, and consequently $\{a,c\}$ is the terminal set of $H$. Since $c$ is a cut vertex in $H$, every $a-c$ path in $H$ must cross $b$, that is 
    \[
    \dist_{\rho}(a,c) = \dist_{\rho}(a,b) + \dist_{\rho}(b,c) = d_1(a,b) + d_2(b,c)
    \]
    and by definition,
    \[
    d_\rho = \tau(\dist_{\rho}(a,c)) = \tau(d_1 + d_2) = d
    \]
    proving the claim.
\end{proof}
\begin{lemma}[Metricity Testing]\label{lem:SP-metricity-testing}
Let $H$ be a parallel composition of $H_1,H_2$ with terminal set $\{s,t\}$. Suppose that for $i\in \{1,2\}$ $(d_i,\lambda_i)$ is a state of $H_i$ realized by $\rho_i$, and that $\rho_i$ is metric on $H_i$. Then 
$\rho = \rho_1\cup \rho_2$ is metric if and only if the realized state $d,\lambda$ satisfied $d \ge \lambda$.
\end{lemma}
\begin{proof}
    $\rho$ is not metric if and only if there exists an edge $e = ab$ such that $\rho(ab) > \dist_\rho(a,b)$. WLOG assume $e\in E(H_1)$. Since $H_1,H_2$ are metric, any $a-b$ shortest path must cross between $H_1,H_2$ through the vertex cut $\{s,t\}$, that is
    \[
    \rho(ab) > \dist_\rho(a,b) = \dist_\rho(a,s) +\dist_\rho(s,t) + \dist_{\rho}(t,b).
    \]
    Rearranging we get
    \[
    \lambda \ge \rho(ab)-  \dist_\rho(a,s) -\dist_\rho(t,b) > \dist_{\rho}(s,t) = d.
    \]
    That is, if $\rho$ is not metric then $\lambda > d$. If $\rho$ is metric, then for all $e = ab$,
    \[
    \rho(ab) \le \dist_\rho(a,b) \le \dist_\rho(a,s) +\dist_\rho(s,t) + \dist_{\rho}(t,b).
    \]
    that is
    \[
    \rho(ab)-  \dist_\rho(a,s) -\dist_\rho(t,b) \le \dist_\rho(s,t) = d
    \]
    Taking the maximum over all $e =ab$ we get
    \[
    \lambda \le d.
    \]
\end{proof}
\begin{corollary}\label{cor:SP-metricity-kept}
    When constructing the table of $H$, a state $(d,\lambda)$ is kept if and only if it is realized by some metric $\rho$.
\end{corollary}
\begin{proof}
    If $H$ is the series composition of $H_1,H_2$, no new cycles are formed in $H$ By induction, since every state of $H_1,H_2$ is metric, so are all states of $H$.

    If $H$ is a parallel composition of $H_1,H_2$, then a state $(d,\lambda)$ is kept if and only if it passes the check $d \ge \lambda$. By \Cref{lem:SP-metricity-testing}, this happens if and only if it is realized by some metric $\rho$.
\end{proof}
\begin{lemma}[DP invariant]
\label{lem:SP-invariant}
    Suppose $A_H$ is the DP table of $H$ after processing all pairs of states $(d_1,\lambda_1),(d_2,\lambda_2)$ for $H_1,H_2$. Suppose $A_{H_1}, A_{H_2}$ are correct for $H_1,H_2$. Then for every state $(d,\lambda)$ of $H$,
    \[
    A_H(d,\lambda) = \min_{\rho\in \Xx(d,\lambda)}\cost(\rho).
    \]
\end{lemma}
\begin{proof}
    Let $(d,\lambda)$ be a valid state of $H$, and let $\DP = A_H(d,\lambda)$. Let $\OPT = \min_{\rho\in \Xx(d,\lambda)}\mathbf{cost}(\rho)$. By \Cref{cor:SP-metricity-kept}, we only keep states realized by metrics, so $\DP \ge \OPT$. Conversely, let $\rho^\star$ be a metric that attains the optimal cost. Its restriction to $H_i$, $\rho_i$ is a valid assignment inducing a state $(d_i,\lambda_i)$. 
    By \Cref{lem:SP-update}, their composition is $(d,\lambda)$. By the correctness of $A_{H_i}$, $A_{H_i}(d_i,\lambda_i) \le \mathbf{cost}(\rho_i)$. Therefore
    \[
    \DP \le A_{H_1}(d_1,\lambda_1) + A_{H_2}(d_2,\lambda_2) \le \mathbf{cost}(\rho_1)+\mathbf{cost}(\rho_2) = \mathbf{cost}(\rho^\star) = \OPT.
    \]
\end{proof}

\begin{theorem}
Integer-valued \IOMR{} on Series-Parallel graphs, with repaired edge lengths restricted
to $[W]$, can be solved in time
\[
O(W^{4}\cdot |E|).
\]
\end{theorem}

\begin{proof}
Let $G$ be a SP graph.
We assume $G$ is $2$ connected. If it is not, we apply the algorithm repeatedly to each $2$-connected block.

Let $T$ be an SP decomposition of $G$. The leaves of $T$ correspond to edges of $G$, and for every $e\in E(G)$, $\EdgeTable(e)$ holds the correct value for \IOMR{} of the subgraph $e$. By inductively applying \Cref{lem:SP-invariant}, the dynamic program of every node is updated correctly. The subgraph represented at the root of $T$ is $G$, so the dynamic program holds the minimum Increase Only Metric Repair cost for $G$.

It remains to bound the running time. A two-terminal state consists of distance values in $\{0,\ldots,W+1\}$ and demand values in
$\{0,\ldots,W\}$, so there are $O(W^{2})$ states. A naive composition of
states considers all pairs of states and performs only constant-size
arithmetic, so each composition takes $O(W^{4})$ time. The tree $T$ has
$O(|E|)$ nodes and edges, and each edge of $T$ participates in
constantly many table operations. Therefore the total running time is $O(W^{4}\cdot |E|)$.

\end{proof}
\paragraph{Extracting a metric repair function}
 To construct a function $\rho$ attaining this cost, one needs to remember which edge guesses attain minimum costs. When processing $T$, if a parallel composition involved a leaf, we keep the ID of this leaf and its guess value alongside the table. This incurs an additional constant space and time cost for each update, so the asymptotic runtime and space complexity stays the same. To extract $\rho$, simply follow these pointers down the trees starting at the root.

\paragraph{Adapting the algorithm to the general case.} To adapt the algorithm to the \MR~problem, one is not restricted to values greater than $w(e)$, but to all valued in $[W]$, and the cost of a metric repair function $\rho$ is the number of modified, increased or decreased, edges. Changing the edge tables to include all values in $[W]$ results in the general algorithm.

%% file: bdd_tw_fpt.tex
In this section we generalize the algorithm from \Cref{sec:sp} to trees of bounded treewidth. We begin by introducing the terminology and notation needed.
\subsection{Set Up}
Let \(G=(V,E)\) be a graph. a \emph{tree decomposition} of $G$ is a pair
\[
    \mathcal T=(T,\{B_t\}_{t\in V(T)})
\]
where $T$ is a tree, and each node \(t\in V(T)\) is assigned a set $B_t\subseteq V(G)$, called the \emph{bag} at \(t\), such that the following properties hold:
\begin{enumerate}
    \item \emph{Vertex coverage:} every vertex of \(G\) appears in some bag:
    \[
        \bigcup_{t\in V(T)} B_t = V(G).
    \]

    \item \emph{Edge coverage:} for every edge \(uv\in E(G)\), there is a bag containing both
    endpoints:
    \[
        \exists t\in V(T) \quad \text{such that} \quad u,v\in B_t.
    \]

    \item \emph{Running intersection:} for every vertex \(v\in V(G)\), the set of tree nodes whose
    bags contain \(v\),
    \[
        \{t\in V(T):v\in B_t\},
    \]
    induces a connected subtree of \(T\).
\end{enumerate}
The \emph{width} of $\Tt$ is 
\[
\mathrm{width}(\Tt) = \max_{t\in V(T)} |B_t| - 1.
\]
The \emph{treewidth} of $G$ is 
\[
\tw(G) = \min_{\Tt}\textrm{width(\Tt)}
\]
where the minimum is taken over all tree decompositions of $G$. 

\paragraph{(Nice) Tree decomposition notation.}

We say that a tree decomposition is \emph{nice}~\cite{bodlaender1996efficient} if every node of \(T\) is one of the following types:
\begin{itemize}
    \item A \emph{leaf} node, whose bag is empty.

    \item An \emph{introduce-vertex} node, with one child \(t'\), where
    \(B_t=B_{t'}\cup\{v\}\) for some vertex \(v\notin B_{t'}\).

    \item A \emph{forget-vertex} node, with one child \(t'\), where
    \(B_t=B_{t'}\setminus\{v\}\) for some vertex \(v\in B_{t'}\).

    \item An \emph{introduce-edge} node, with one child \(t'\), where \(B_t=B_{t'}\), and one edge
    \(uv\in E(G)\), with \(u,v\in B_t\), is introduced.

    \item A \emph{join} node, with two children \(t_1,t_2\), where
    \(B_t=B_{t_1}=B_{t_2}\).
\end{itemize}
It is known~\cite{bodlaender1996efficient} that given $G$ of treewidth $r$, one can find a nice tree decomposition in $O(f(r)\cdot |E|)$. From this point onwards, we assume we have access to a nice tree decomposition of $G$. This will make the Dynamic Program cleaner. To set notation, for a tree node $t$, let $G_t = (V(G_t), E_t)$ be the subgraph induced by the subtree rooted in $t$. In DP terms,$G_t$ is the graph processed by the DP up to node $t$. Let $\overline{G}_t$ be the graph left to process, alongside $B_t$. That is, $V(G_t)\cap V(\overline{G}_t) = B_t$.
Note that if $t$ is a child of $t'$, then $B_t$ is a vertex separator in $G_{t'}$. This is the analog of a two-terminal set for Series Parallel graphs, we therefore call the vertices in $B_t$ the \emph{terminals}.

\subsection{Bag profiles}

Recall from \Cref{sec:sp} that
\[
\tau(k):=\min\{k,\,W+1\},\qquad [k]_+:=\min\{W,\max\{0,k\}\},
\]
We describe the dynamic program's states and tables in an analogous way to \Cref{sec:sp}. 

\paragraph{Leaf tables.} The base case is a leaf node, whose bag is empty and whose graph $G_t$ has no
edges. Its table holds the single empty profile with cost $0$.

\paragraph{Bag profiles.} Let $t$ be a node with terminals $B_t$ and processed graph $G_t=(V_t,E_t)$. Recall that $\dist_{\rho,G_t}$ is the distance function induced by $\rho$ on the graph $G_t$. We write $\dist_t\defeq\dist_{\rho,G_t}$. For an increase-only repaired weight function $\rho:E_t\to[W]$, the
\emph{bag profile realized by $\rho$} is the pair of matrices $(D_\rho,\Lambda_\rho)$, indexed by $B_t$,
defined for $y,z\in B_t$ by
\[
D_\rho(y,z)=\tau\big(\dist_t(y,z)\big),
\qquad
\Lambda_\rho(y,z)=\psbk{\max_{e=ab\in E_t}\big(\rho(e)-\dist_t(a,y)-\dist_t(b,z)\big)} ,
\]
where the maximum ranges over both orientations of each edge. Thus $D_\rho(y,z)\in\{0,\dots,W+1\}$ and
$\Lambda_\rho(y,z)\in\{0,\dots,W\}$. Intuitively, $\Lambda_\rho(y,z)$ records how long a future outside
route between $y$ and $z$ must be in order not to make an already-processed edge of $G_t$ heavy---the
matrix version of $\lambda_\rho$ in \Cref{sec:sp}.

We say a state $(D_t,\Lambda_t)$ is \emph{realized by $\rho$} if $D_t=D_\rho$ and $\Lambda_t=\Lambda_\rho$.
Let $\Xx(D_t,\Lambda_t)$ be the set of all increase-only functions $\rho:E_t\to[W]$ that are metric
on $G_t$ and realize $(D_t,\Lambda_t)$. The dynamic-programming table at $t$ stores
\[
A_t(D_t,\Lambda_t)=\min_{\rho\in\mathcal X(D_t,\Lambda_t)}\cost(\rho),
\]
and $A_t(D_t,\Lambda_t)=+\infty$ when $\mathcal X(D_t,\Lambda_t)=\emptyset$. The program never stores a
witnessing $\rho$; every finite entry should be read as being represented by at least one such function,
an invariant maintained by the transitions and we prove in \Cref{lem:tw-inv}.

\subsection{Projection of demands}

When we build $G_t$ by introducing an edge or by joining two children, the bag $B_t$ is unchanged, but
the distances between its vertices may shrink and edges processed earlier may constrain future routes --- these are the analogous operations to parallel and series composition. We isolate this bookkeeping in a single \emph{projection} operation---the treewidth analogue of the demand projection $\Pi_{H_i\to S}$. In both these cases, $B_t$ acts as a vertex separator between $G_t$ and the subgraphs induced by the children of $t$, hence we project the demands induced by $\Lambda$ onto pairs of terminals in $B_t$, according to the new distance matrix $D$ of $G_t$.

Recall that the demand realized by an increase-only assignment $\rho$ on a processed subgraph $G'=(V',E')$ with bag $B$ is the matrix indexed by $p,q\in B$ defined by 
\[
\Lambda'(p,q)=\psbk{\max_{e=ab\in E'}\big(\rho(e)-\dist_{\rho,G'}(a,p)-\dist_{\rho,G'}(b,q)\big)},
\]
and the maximum ranging over both orientations of each edge.

\begin{definition}[Projection]
Let $B$ be a set of bag vertices, and let $D$ be a distance matrix and $\Lambda$ a demand matrix, both
indexed by $B$. The \emph{projection of $\Lambda$ through $D$} is the demand matrix $\Pi_D[\Lambda]$ with
\[
\Pi_D[\Lambda](y,z)=\psbk{\max_{p,q\in B}\big(\Lambda(p,q)-D(p,y)-D(q,z)\big)}
\]
\end{definition}

\begin{lemma}[Projection of demands]\label{lem:tw-proj}
Let $G'=(V',E')$ be a subgraph of $G_t$ with $B\subseteq V'$ a vertex separator between $G'$ and the rest
of $G_t$, let $\rho$ be increase-only on $E_t$, and let $\Lambda'$ be the demand realized by $\rho|_{E'}$
on $G'$. Let $D_t$ be the truncated distance matrix of $G_t$ restricted to $B$. Then for all $y,z\in B$,
\[
\Pi_{D_t}[\Lambda'](y,z)=\Big[\max_{e=ab\in E'}\big(\rho(e)-\dist_{\rho,G_t}(a,y)-\dist_{\rho,G_t}(b,z)\big)\Big]_+ ,
\]
the maximum ranging over both orientations of each edge.
\end{lemma}

\begin{proof}
Write $\dist_t=\dist_{\rho,G_t}$ and $\dist'=\dist_{\rho,G'}$. Fix an edge $ab\in E'$ and $y\in B$. Since
$B$ separates $G'$ from the rest of $G_t$, every $a$--$y$ path in $G_t$ crosses $B$; up to its first vertex
of $B$ the path stays in $G'$ and is shortest there. Hence
\[
\dist_t(a,y)=\min_{p\in B}\big(\dist'(a,p)+D_t(p,y)\big)
\]
and similarly
\[
\dist_t(b,z)=\min_{q\in B}\big(\dist'(b,q)+D_t(q,z)\big)
\]
Substituting and turning the two minima into a joint maximum,
\[
\rho(ab)-\dist_t(a,y)-\dist_t(b,z)
=\max_{p,q\in B}\big(\rho(ab)-\dist'(a,p)-\dist'(b,q)-D_t(p,y)-D_t(q,z)\big).
\]
Recall that the \emph{unclamped} demand on $G'$ is
\[
\max_{ab\in E'}\big(\rho(ab)-\dist'(a,p)-\dist'(b,q)\big) = \widetilde\Lambda'(p,q).
\]
So taking the maximum over all $ab\in E'$ and exchanging the two maxima we get
\[
\max_{ab\in E'}\big(\rho(ab)-\dist_t(a,y)-\dist_t(b,z)\big)
=\max_{p,q\in B}\Big(\widetilde\Lambda'(p,q)-D_t(p,y)-D_t(q,z)\Big).
\]

Edge weights are at most $W$ and internal distances are nonnegative, so
$\widetilde\Lambda'(p,q)\le W$; and since $D_t(p,y),D_t(q,z)\ge0$, any pair with
$\widetilde\Lambda'(p,q)<0$ contributes a nonpositive term and is therefore irrelevant after the outer
clamp. Thus replacing $\widetilde\Lambda'$ by $\Lambda'=[\widetilde\Lambda']_+$ leaves the value unchanged,
and applying $[\cdot]_+$ to both sides yields
\[\psbk{\max_{ab\in E'}\big(\rho(ab)-\dist_t(a,y)-\dist_t(b,z)\big)}
=\psbk{\max_{p,q\in B}\big(\Lambda'(p,q)-D_t(p,y)-D_t(q,z)\big)}
\]
and the right hand side is precisely $\Pi_{D_t}[\Lambda'](y,z)$, concluding the proof.
\end{proof}

\begin{corollary}[Demand of a composition]\label{cor:tw-demand-composition}
Suppose $E_t=E'\sqcup E''$, and let $\Lambda',\Lambda''$ be the demands realized by $\rho$ on the two
parts. Then the demand realized on $G_t$ is the entrywise maximum of the projected parts,
\[
\Lambda_t=\Pi_{D_t}[\Lambda']\ \vee\ \Pi_{D_t}[\Lambda''].
\]
where $\vee$ is the entry-wise maximum of the two matrices.
\end{corollary}

\begin{proof}
By definition $\Lambda_t(y,z)=\big[\max_{e\in E_t}(\rho(e)-\dist_t(a,y)-\dist_t(b,z))\big]_+$, and since
$E_t=E'\sqcup E''$ the maximum over $E_t$ is the larger of the maxima over $E'$ and $E''$. Applying
\Cref{lem:tw-proj} to each part gives the claim.
\end{proof}

\subsection{Transition Rules}
We now describe the table update at each node type.

Introduce Vertex and Forget Vertex nodes add no edges to $G_t$ and close no cycles, while Introduce Edge and Join nodes may close cycles. They are treated in a similar manner to the parallel composition in \Cref{sec:sp}: project the demands from the previously processed subgraphs onto the bag $B_t$ according to the new distance matrix $D_t$, then test for metricity.

\paragraph{Introduce Vertex} If a vertex $v$ is added to $B_{t-1}$ to form $B_t$, it has no incident edges in $G_t$. For every state $(D_{t-1},\Lambda_{t-1})$ of finite cost in $B_{t-1}$:
\begin{enumerate}
    \item Extend the matrix $D_{t-1}$ to $D_t$ by setting $D_t(v,v) = 0$. For every $u\in B_{t-1}$ set $D_t(u,v) = D_t(v,u) = W+1$ and $\Lambda_t(u,v) = \Lambda_t(v,u) = 0$.
    \item Set $A_t(D_t,\Lambda_t) = A_{t-1}(D_{t-1},\Lambda_{t-1})$.
\end{enumerate}
An isolated terminal is at (truncated) distance $W+1$ from every other terminal and imposes no demand, so no metricity test is needed.

\paragraph{Forget vertex} Suppose $t$ forgets vertex $v$ with child $t-1$. For every valid state $(D_{t-1},\Lambda_{t-1})$:
\begin{enumerate}
    \item Form $(D_t,\Lambda_t)$ by deleting the row and column corresponding to $v$.
    \item Update $A_t(D_t,\Lambda_t) = \min\Big(A_t(D_t,\Lambda_t),A_{t-1}(D_{t-1},\Lambda_{t-1})\Big)$
\end{enumerate}
Again, no metricity test is needed.
\paragraph{Introduce Edge} Introducing edge $e = uv$ of weight $w(e)$. For every valid state $(D_{t-1},\Lambda_{t-1})$ and every possible value $\rho(e)\in \{w(e),w(e) +1,\ldots ,W\}$:
\begin{enumerate}
    \item If $\rho(e) > D_{t-1}(u,v)$, discard the guess
    \item For all $y,z\in B_t$, set
    \begin{flalign*}
    \Dnew(y, z) = \tau\rbk{\min \Big\{D_{t-1}(y,z),\rho(uv) + D_{t-1}(y,u) + D_{t-1}(v,z)\Big\}}
    \end{flalign*}
    where the minimum is taken over both orientations $uv, vu$.

    \item Project the old demands through the new distances and add the demand of $e$:
    \[
    \Lnew = \Pi_{\Dnew}[\Lambda_{t-1}] \vee \Lambda_e
    \]
    Where 
    \[
    \Lambda_e (y,z) = \psbk{\max_{(a,b)\in \{(u,v),(v,u)\}}(\rho(e) - \Dnew(a,y)-\Dnew(b,z)}
    \]
    \item If $\Dnew(y,z) < \Lnew(y,z)$ for some $y,z\in B_t$, discard the guess. 
    \item Otherwise update $A_t(\Dnew, \Lnew) = \min\Big(A_t(\Dnew,\Lnew), A_{t-1}(D_{t-1},\Lambda_{t-1}) + c_e(\rho(e))\Big)$.
    \end{enumerate}
\paragraph{Join} Suppose $t$ is a join of children $t_1,t_2$ and $B_t = B_{t_1} = B_{t_2}$. For every pair of states $(D_1,\Lambda_1)$ of $t_1$ and $(D_2,\Lambda_2)$ of $t_2$:
\begin{enumerate}
    \item Set $\Dtemp = D_1(y,z)\wedge D_2(y,z))$, where $\wedge$ is the entrywise minimum of $D_1,D_2$. Let $D_t$ be the truncated all-pairs-shortest-paths closure of $\Dtemp$ over $B_t$.
    \item project the merged demands $\Lambda_1,\Lambda_2$ through $D_t$:
    \[
    \Lambda_t = \Pi_{D_t}[\Lambda_1\vee \Lambda_2].
    \]
    \item If $D_t(y,z) < \Lambda_t(y,z)$ for some $y,z\in B_t$, discard the pair $(D_t,\Lambda_t)$. 
    \item Otherwise update $A_t(D_t,\Lambda_t) = \min\Big( A_t(D_t, \Lambda_t), \; A_{t_1}(D_1, \Lambda_1) + A_{t_2}(D_2, \Lambda_2) \Big) $.
\end{enumerate}

\subsection{Correctness}
We prove correctness by showing that at every node $t$ and every state $(D_t,\Lambda_t)$ the invariant
\begin{equation}\tag{$\star$}\label{eq:inv}A_t(D_t,\Lambda_t) = \min_{\rho\in \Xx(D_t,\Lambda_t)}\cost(\rho)
\end{equation}
is maintained.
Leaf and Introduce/Delete vertex nodes add no edge and leave the metric structure of the processed graph unchanged, so they are immediate. The substance is in the two cycle-closing nodes, Introduce Edge and Join, which we treat in parallel and distinguish only where necessary. For ease of presentation, for an Introduce Edge and Join nodes we write  $E_t = E_{t_1}\sqcup E_{t_2}$. For a Join node, $E_{t_i}$ is the edge set of the graph processed until the child $t_i$. For Introduce Edge, we think of $E_{t_1} = E_{t-1}$, that is the edge set of the graph $G_{t-1}$ processed until the child $t-1$ of $r$, and $E_{t_2} = \{e^*\}$ is the newly introduced edge. We furthermore write the combined assignment $\rho = \rho_1 \cup \rho_2$, where $\rho_1,\rho_2$ are metrics realizing the states of $t_1,t_2$ respectively. For an Introduce Edge node, $\rho_2$ is simply the function $e^*\mapsto \rho(e^*)$.

The proof follows the structure of the proof of \Cref{sec:sp}. We first show that the transitions compute the correct updated profiles (\Cref{lem:tw-correct-update}). Then we show that a state is valid if and only if it is realized by some metric $\rho$ (\Cref{lem:tw-metr-test} and \Cref{cor:tw-real-states}). Finally, we show that we maintain the DP invariant above in \Cref{lem:tw-inv}, and conclude with the main theorem. 

Throughout we assume the children's tables satisfy \eqref{eq:inv}; in particular every child state we combine is realized
by some increase-only $\rho_i$ that is metric on the child's subgraph. We write $\dist_t\defeq\dist_{\rho,G_t}$,
and note that truncating distances at $W+1$ and clamping demands into $[0,W]$ never affects metricity: a
broken cycle has a heavy edge of weight at most $W$ and a complementary path of smaller weight, so all
quantities that matter lie in the tracked ranges.

\begin{lemma}[Correctness of updates]\label{lem:tw-correct-update}
    Let $t$ be an introduce-edge node adding $e^*=uv$ with a chosen $\rho(e^*)\le D_{t-1}(u,v)$, or a join node,
and let $\rho$ be the combined assignment of child assignments realizing the combined child states. Then the
state $(D_t,\Lambda_t)$ produced by the transition is realized by $\rho$: that is, $D_t=D_\rho$ and
$\Lambda_t=\Lambda_\rho$.
\end{lemma}

\begin{proof}
    We first show that $D_t = D_\rho$. For an Introduce Edge node, a shortest $y-z$ path in $G_t$ either uses $e^*$ once, or avoids it altogether. If it avoids $e^*$, it lies entirely in $G_{t-1}$ and has length $\dist_{t-1}(y,z)$. If it uses $e^*$, it splits into two paths lying in $G_{t-1}$. Therefore
\[
\dist_t(y,z)=\min\big\{\dist_{t-1}(y,z),
\rho(e^*)+\dist_{t-1}(y,u)+\dist_{t-1}(v,z),
\rho(e^*)+\dist_{t-1}(y,v)+\dist_{t-1}(u,z)\big\}.
\]
By the inductive hypothesis, $\dist_{t-1}$ agrees with $D_{t-1}$ on $B_t$, so 
\[
\dist_t(y,z)=\min\big\{D_{t-1}(y,z),
\rho(e^*)+D_{t-1}(y,u)+D_{t-1}(v,z),
\rho(e^*)+D_{t-1}(y,v)+D_{t-1}(u,z)\big\}.
\]
Applying $\tau$ to both sides we get 
\[
D_\rho(y,z) =\tau(\dist_t(y,z)) = \Dnew(y,z).
\]

For a join node, $B_t$ separates $V_{t_1}$ from $V_{t_2}$, so a shortest $y-z$ path in $G_t$ splits into
maximal subpaths, each lying in one child and joining two terminals $p,q\in B_t$; being a subpath of a
shortest path it is itself shortest in that child, of length $\dist_{t_i}(p,q)=D_i(p,q)$ by the inductive
hypothesis. The terminals visited thus form a walk over $B_t$ with edge weights $\Dtemp=D_1\wedge D_2$,
of total length $\dist_t(y,z)$; conversely any such walk realizes a path of that length. Hence $\dist_t(y,z)$
equals the all-pairs shortest-path closure of $\Dtemp$, i.e. $D_t(y,z)=\tau(\dist_t(y,z))=D_\rho(y,z)$.\\

We now show that the demand updates are correct. In both cases $E_t$ is the disjoint union of the parts carried by the children, and by the
distance computation the transition's distance matrix is the true $D_\rho$. \Cref{cor:tw-demand-composition} then gives
$\Lambda_\rho=\Pi_{D_\rho}[\Lambda']\vee\Pi_{D_\rho}[\Lambda'']$, where $\Lambda',\Lambda''$ are the demands
realized by $\rho$ on the two parts. For a Join node these are $\Lambda_1,\Lambda_2$. Since $\max$ distributes over $\max$, the projection distributes over the entrywise maximum, we have
\[
\Lambda_t = \Pi_D[\Lambda'\vee\Lambda'']=\Pi_D[\Lambda']\vee\Pi_D[\Lambda''] = \Lambda_\rho.
\]
For an Introduce-Edge node $\Lambda'=\Lambda_{t-1}$. $B_t$ trivially separates $G_{t-1}$ from the subgraph on the newly introduced edge $e^*$, so by \Cref{lem:tw-proj}
\[
\Pi_{D_\rho}[\Lambda'']=\Lambda_{e}.
\] 
Hence
\[
\Lnew = \Pi_{\Dnew}[\Lambda_{t-1}] \vee \Lambda_e = \Pi_{\Dnew}[\Lambda_{t-1}\vee \Lambda_e] = \Pi_{\Dnew}[\Lambda' \vee \Lambda''] =  \Lambda_\rho.
\]
\end{proof}

\begin{lemma}[Metricity testing]\label{lem:tw-metr-test}
Let $t$ be an introduce-edge or join node, and let $\rho$ be the combined assignment of \Cref{lem:tw-correct-update}. Suppose that the assignment on each child is metric on its subgraph. For an introduce-edge node suppose further that $\rho(e^*)\le D_{t-1}(u,v)$.
Then $\rho$ is metric on $G_t$ if and only if $D_\rho(y,z)\ge\Lambda_\rho(y,z)$ for all $y,z\in B_t$.
\end{lemma}
\begin{proof}
    Suppose $\rho$ is metric on $G_t$. Then for every edge $ab\in E_t$ and all $y,z\in B_t$, the triangle inequality gives $\rho(ab)\le\dist_t(a,b)\le\dist_t(a,y)+\dist_t(y,z)+\dist_t(z,b)$, hence $\rho(ab)-\dist_t(a,y)-\dist_t(z,b)\le\dist_t(y,z)$. Taking the maximum over edges and orientations, $\Lambda_\rho(y,z)\le\dist_t(y,z)$, and after truncation $\Lambda_\rho(y,z)\le D_\rho(y,z)$.\\
    
    Conversely, suppose $\rho$ is not metric, and let $e=ab$ be a heavy edge, so a shortest
$a$--$b$ path $\pi$ has length $\dist_t(a,b)<\rho(ab)$. Without loss of generality $ab$ is carried by the first
child (or by $G_{t-1}$, for an introduce-edge node). That child is metric, so $\pi$ cannot lie inside it.
For a Join node, $\pi$ therefore uses an edge of $G_{t_2}$. Consequently, $\pi$ crosses $B_t$: let $y$ be the first vertex of
$B_t$ on $\pi$ (starting from $a$) and $z$ the last (before $b$).
The subpaths $a\to y$ and $z\to b$ avoid the
crossing edge and lie in the first child, hence have length at least $\dist_t(a,y)$ and $\dist_t(z,b)$, while
the middle subpath has length at least $\dist_t(y,z)$. Therefore
\[
\rho(ab)>\dist_t(a,b)=|\pi|\ge\dist_t(a,y)+\dist_t(y,z)+\dist_t(z,b),
\]
so
\[
\Lambda_\rho(y,z)\ge\rho(ab)-\dist_t(a,y)-\dist_t(z,b)>\dist_t(y,z)=D_\rho(y,z),
\] 
and the check fails.

For an Introduce Edge node, $\pi$ lies entirely in $G_{t-1}$ except for a single traverse of $e^* = uv$ which is not heavy Since $\rho(e^*)\le D_{t-1}(u,v)$. Since $u,v\in B_t$, again $\pi$ must cross $B_t$ and a similar argument holds.
\end{proof}

\begin{corollary}[The DP keeps only realizable states]\label{cor:tw-real-states}
When the transition at a node $t$ processes its inputs, a state $(D_t,\Lambda_t)$ is assigned a finite value
if and only if it is realized by some increase-only $\rho$ that is metric on $G_t$. In particular, the
program never stores a state that is not realized by an increase-only metric repair of $G_t$.
\end{corollary}
\begin{proof}
    For leaf and vertex nodes, $E_t$ and the metric structure of the processed graph are unchanged: by the
inductive hypothesis every child state is realized by a metric $\rho$, the transition copies it (for an
introduce-vertex node the added terminal is isolated, at truncated distance $W+1$ from every other and
imposing no demand), and $\rho$ remains metric on $G_t$. No state is discarded, and every kept state is
realized by a metric $\rho$.

For an Introduce Edge or Join node, the transition keeps a state $(D_t,\Lambda_t)$ exactly when
the guess or pair passes the check $D_t\ge\Lambda_t$ (and, for an introduce-edge node, $\rho(e^*)\le D_{t-1}(u,v)$).
\begin{itemize}
\item If the state is kept, then by \Cref{lem:tw-correct-update} it is realized by the combined assignment $\rho$. Its child assignments are metric on their subgraphs by the inductive hypothesis, so by \Cref{lem:tw-metr-test} $\rho$ is metric on $G_t$. Thus the kept state is realized by a metric $\rho$.
\item Conversely, suppose $(D_t,\Lambda_t)$ is realized by some increase-only $\rho$ that is metric on $G_t$.
Restricting $\rho$ to each child's edge set gives increase-only assignments that are metric on the children.
Passing to a subgraph only increases distances, so $\rho(ab)\le\dist_{G_t}(a,b)\le\dist_{G_{t_i}}(a,b)$ for
every child edge $ab$. By the inductive hypothesis these restrictions realize valid child states.
The transition considers such states as it iterates over all valid states. By \Cref{lem:tw-correct-update} it produces $(D_t,\Lambda_t)$,
and by \Cref{lem:tw-metr-test} the check passes. Hence the state is kept.
\end{itemize}
Therefore a state is kept if and only if it is realized by some increase-only metric $\rho$.
\end{proof}

\begin{lemma}[DP invariant]\label{lem:tw-inv}
Every node $t$ satisfies the invariant \eqref{eq:inv}.
\end{lemma}
\begin{proof}
    We induct on the decomposition tree.

\paragraph{Leaf.} The bag is empty and $G_t$ has no edges, so the only state is the empty profile, realized solely
by the empty assignment of cost $0$; the table stores $0$.

\paragraph{Introduce/Forget vertex.} Since $E_t=E_{t-1}$, the invariant is maintained trivially.

\emph{Introduce edge / Join.} Fix a state $(D_t,\Lambda_t)$ and write
\[
\DP=A_t(D_t,\Lambda_t),\qquad \OPT=\min_{\rho\in\Xx(D_t,\Lambda_t)}\cost(\rho).
\]

By \Cref{cor:tw-real-states}, the program only keeps valid states realized by some metric $\rho\in\Xx$, so $\DP\ge \OPT$.
Conversely, let $\rho^\star\in\Xx(D_t,\Lambda_t)$
attain $\OPT$. Its restrictions $\rho^\star_i$ are metric on the children and realize child states with
$A_{t_i}(\cdot)\le\cost(\rho^\star_i)$ by the inductive hypothesis. The transition processes these states. By \Cref{lem:tw-correct-update}, it yields $(D_t,\Lambda_t)$, and by \Cref{lem:tw-metr-test} the check
passes, and it updates
\[
\DP\le A_{t_1}(\cdot)+A_{t_2}(\cdot)
\le \cost(\rho^\star_1)+\cost(\rho^\star_2) =\cost(\rho^\star)=\OPT,\]
where for Introduce Edge node we think of $A_{t_2}$ as the cost of $\rho(e^*)$.
Therefore $\DP=\OPT$.
\end{proof}

We conclude with the main theorem:
\begin{theorem}\label{thm:tw-algorithm}
    Integer-valued IOMR on graphs of treewidth at most $r$, with repaired edge lengths restricted to
$[W]$, can be solved in time
\[
O\Big(W^{O(r^2)}\cdot (|V|+|E|)\Big)
\]
given a width $r$ nice tree decomposition.  In particular, the problem is fixed-parameter
tractable parameterized by $(r,W)$.
\end{theorem}
\begin{proof}
By \Cref{lem:tw-inv}, every table entry satisfies the invariant \eqref{eq:inv}. The root has an empty bag, so $\Xx(D_\emptyset,\Lambda_\emptyset)$ is the set of all metric increase only functions $\rho:E(G)\to [W]$, and the program returns the correct minimum value. 

It remains to show the runtime. A bag has at most $r+1$ vertices, so a state consists of $O((r+1)^2)$ entries, each a distance in $\{1,\ldots, W+1\}$ or a demand in $\{0,1,\ldots W\}$. There are thus $O\rbk{(W+1)^{O(r+1)^2}} = W^{O(r^2)}$ states per node. Introduce and Forget Vertex nodes updates take $O(r^2)$ time per state as they are simply matrix update or deletion. An Introduce Edge node tries at most $W$ guesses per state, each costing $O(r^2)$ per matrix. A Join node combines $\rbk{W^{O(r^2)}}^2$ pairs, each require computing a shortest path closure at cost $O(r^3)$ and $O(r^4)$ matrix work. These costs are subsumed by $W^{O(r^2)}$. Summed over all $O(|V| + |E|)$ nodes of the tree decomposition, the running time is 
\[
O\Big(W^{O(r^2)}\cdot (|V|+|E|)\Big).
\]
\end{proof}
\paragraph{Extracting the solution}
The program computes the optimal support for an increase-only repair function. To extract one such function, one can adapt the algorithm in the following way. When processing a node and state $(D_t,\Lambda_t)$, keep a record of which states of children induce it. Moreover, when processing an Introduce Edge node with new edge $e^*$, at state $(D_t,\Lambda_t)$, the table $A_t$ also keeps track of some guess $\rho(e^*)$ attaining the optimal value stored in $A_t(D_t,\Lambda_t)$. After running the program, the minimum entry at the root's table holds the correct optimum. Chase down the tree the recorded states obtaining this minimum, and when reading Introduce Edge nodes, read the values recorder for the edges and construct $\rho$ realizing an optimal state. These changes cost an additional $W^{O(r^2)}$ space, and one more traversal of the tree -- so the runtime and space complexity stay the same.
\paragraph{Adapting the algorithm to the general case.}
Much like the algorithm on Series Parallel graphs, when processing an Introduce Edge node we do not restrict $\rho(e^*) \ge w(e)$ and adjust the cost function accordingly. The transition rules stay the same.

\subsection{Application for length-bounded multicut}
Recall that in the \emph{length-bounded multicut} problem we are given a graph $G=(V,E)$,
a collection of terminal pairs $(s_i,t_i)_{i=1}^{k}$, and a bound $L\in\NN$,
and we seek a smallest set $S\subseteq E$ such that no $s_i$--$t_i$ path of
length at most $L$ survives in $G\setminus S$, i.e.\ $\dist_{G\setminus
S}(s_i,t_i)\ge L+1$ for all $i$. We now show that the algorithm of
\Cref{thm:tw-algorithm} solves this problem, and give sufficient conditions
under which the dependency on the number of terminals $k$ is eliminated in
the exponent.

The reduction augments $G$ with a chord $s_it_i$ for each pair, and the
running time depends on the treewidth of the resulting graph. We first bound
this treewidth: it suffices that the pairs be joined by short paths of low
congestion.
\begin{lemma}\label{lem:tw-dist-cong-bound}
Let $G=(V,E)$ have a tree decomposition of width $f$. Let
$(s_i,t_i)_{i=1}^{k}$ be pairs, and suppose there exist $s_i$--$t_i$ paths
$\pi_1,\dots,\pi_k$ in $G$ such that each $\pi_i$ has at most $c$ edges and
each vertex of $G$ lies on at most $d$ of the paths. Then
$G'=(V,E\cup\{s_it_i:i\in[k]\})$ satisfies
\[
  \tw(G')\le (f+1)(1+cd)-1 = O(cd\cdot f).
\]
\end{lemma}
\begin{proof}
Fix a width-$f$ tree decomposition $\Tt=(T,\{B_t\})$ of $G$. We build
a decomposition of $G'$ by enlarging bags, then bound the enlargement.
For each pair $i$, the path $\pi_i$ uses at most $c+1$ vertices. By edge
coverage and running intersection, the nodes of $T$ whose bags meet $\pi_i$
form a connected subtree $T_i$, and consecutive edges of $\pi_i$ are covered
by adjacent bags. Therefore adding the $\le c+1$ vertices of $\pi_i$ to every bag of
$T_i$ yields a valid decomposition in which $s_i$ and $t_i$ are co-bagged
throughout $T_i$, so the chord $s_it_i$ is covered. Doing this for all $i$
simultaneously gives a decomposition $\mathcal{T}'$ of $G'$: coverage of the
new edges holds by construction, and each original property is preserved
since we only added vertices to bags along connected subtrees.
It remains to bound $|B'_t|$. A vertex $v$ is added to $B_t$ only on account
of some path $\pi_i$ with $v\in\pi_i$ and $t\in T_i$. Fix $t$. Each vertex
$u\in B_t$ belongs to at most $d$ paths, and each such path contributes at
most $c$ other vertices to $B'_t$. Hence the number of added vertices is at
most $|B_t|\cdot d\cdot c\le(f+1)cd$, so
\[
  |B'_t|\ \le\ (f+1)+(f+1)cd\ =\ (f+1)(1+cd).
\]
Therefore $\tw(G')\le(f+1)(1+cd)-1=O(cd\cdot f)$.
\end{proof}

\begin{corollary}[Length-bounded multicut]\label{cor:tw-lbmc}
Let $G=(V,E)$ be a graph, $(s_i,t_i)_{i=1}^{k}$ a collection of terminal
pairs, and $L\in\NN$. Let $G'=(V,E\cup\{s_it_i:i\in[k]\})$ and
$r=\tw(G')$. Then a smallest set $S\subseteq E$ with
$\dist_{G\setminus S}(s_i,t_i)\ge L+1$ for all $i$ can be
computed in time
\[
  O\left(L^{O(r^2)}\cdot(|V|+|E|+k)\right).
\]
Moreover $r\le f+k$, where $f=\tw(G)$ and $r=O(f)$ whenever the pairs satisfy the hypotheses of \Cref{lem:tw-dist-cong-bound}
with constant $c,d$---in particular when each pair $s_it_i$ is co-bagged in a
width-$f$ decomposition.
\end{corollary}
\begin{proof}
Weight $G'$ by giving every original edge $e\in E$ weight $1$ and every
added edge $s_it_i$ weight $L+1$, and set $W=L+1$. The broken cycles of $G'$ are exactly cycles that decompose into a pair of terminals $(s_i,t_i)$ and a complementary path $P$ of length at most $L$.
By \Cref{thm:tw-algorithm}, a set of edges $S$ that is a light hitting set for the broken cycles can be computed in $O\left(L^{O(r^2)}\cdot(|V|+|E|+k)\right)$ time. Note that an edge of weight $L+1$ is never light in a cycle of $G'$, as a cycle with two heaviest edges is not broken, so the computed set $S$ is a subset of $E(G)$. In the graph $G\setminus S$, no paths of length at most $L$ between $s_i-t_i$ exist, and therefore $\dist_{G\setminus S}(s_i,t_i) \ge L+1$ as required. The bound $r\le f+k$ follows by adding all $k$ terminals to every bag of a width-$f$ decomposition; $r=O(f)$ under the hypotheses of \Cref{lem:tw-dist-cong-bound} is immediate.
\end{proof}

%% file: bdd-tw-is-hard.tex
\label{sec:bdd-tw-hard}
We now show that, when the weights can be arbitrary rational numbers encoded in binary, metric repair is hard \emph{even} on graphs of bounded treewidth.

\begin{theorem}\label{thmBoundedTreewidthHardness}
    \MR{} and \IOMR{} are both $\NPh$ even on graphs of pathwidth at most 6.
\end{theorem}

\begin{proof}
    We reduce from the \Partition{} problem, which we define as follows:

    \begin{problem}[\Partition]
    	Given a list of $n$ numbers $x_1, x_2, \dots, x_n$, each strictly between 0 and $\frac12$ and summing to 1, does there exist a partition of $[n]$ into disjoint subsets $[n] = P_1 \cup P_2$ such that $\sum_{i \in P_1} x_i = \sum_{i \in P_2} x_i = \frac12$?
    \end{problem}

    \ipncm{.8}{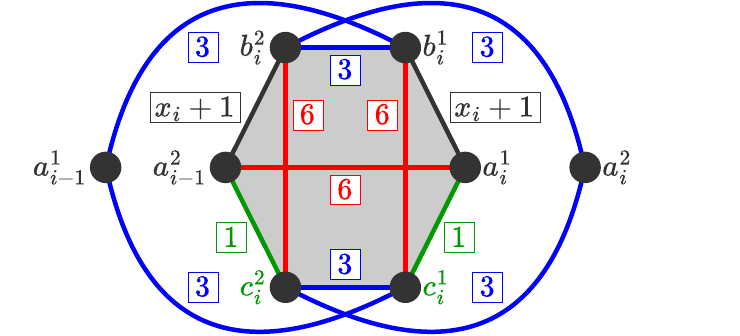}{\label{figBoundedTreewidthHardnessGadget} The gadgets (i.e., weighted subgraphs of $H$) produced in the reduction in the proof of Theorem \ref{thmBoundedTreewidthHardness}, for each $i \in [n]$.}

    This problem is obviously still $\NPh$ with the restrictions that each $0 < x_i < \frac12$. Given a \Partition{} instance $x_1, x_2, \dots, x_n$, we construct a weighted graph $(H, w)$ with a copy of the gadget pictured in Figure~\ref{figBoundedTreewidthHardnessGadget} for each $i \in [n]$, which has the following edge weights:
    \begin{align*}
        w(a_{i - 1}^1b_i^1) = w(b_i^1b_i^2) = w(b_i^2a_i^2) = w(a_{i - 1}^1c_i^1) = w(c_i^1c_i^2) = w(c_i^2a_i^2) &= 3\\
        w(b_i^1c_i^1) = w(b_i^2c_i^2) = w(a_{i - 1}^2a_i^1) &= 6\\
        w(a_{i - 1}^2b_i^2) = w(b_i^1a_i^1) &= x_i + 1\\
        w(a_{i - 1}^2c_i^2) = w(c_i^1a_i^1) &= 1
    \end{align*}
    These gadgets are stitched together by the $a$ vertices with shared names. Finally, we add two additional edges, with weights
    $$w(a_0^1a_n^1) = w(a_0^2a_n^2) = 4n + \frac12.$$
    The entire weighted graph $H$ is shown in Figure~\ref{figBoundedTreewidthHardnessExample}; note that the pathwidth of this graph is 6 for $n \geq 2$.

    \ipncm{.4}{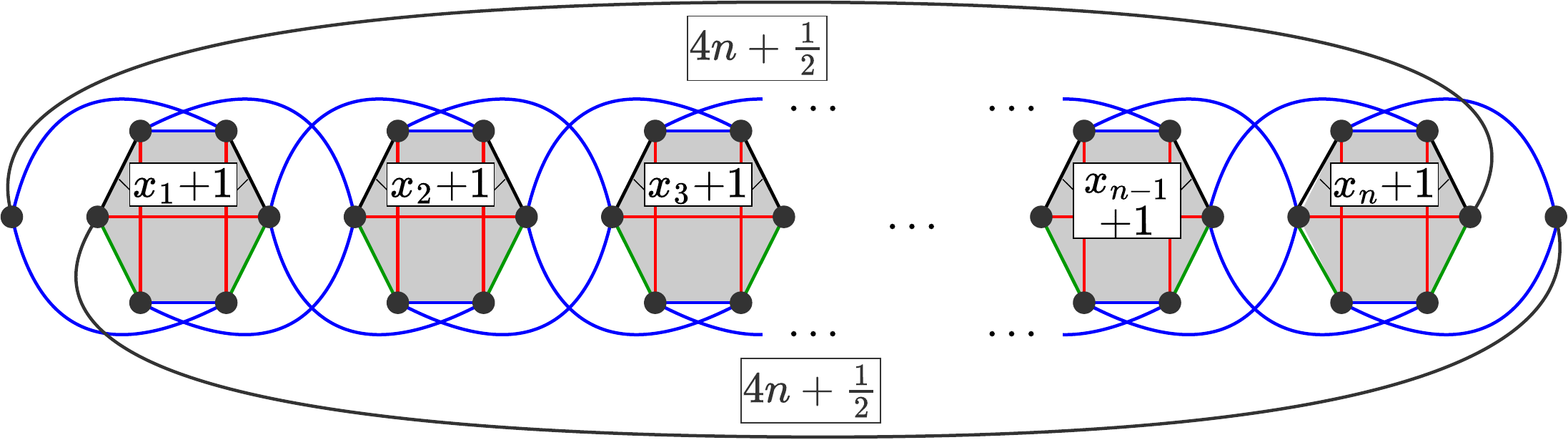}{\label{figBoundedTreewidthHardnessExample} The weighted graph $H$ produced in the reduction from \Partition. As in Figure~\ref{figBoundedTreewidthHardnessGadget}, green edges have weight 1, blue edges have weight 3, and red edges have weight 6, with the weights on black edges labeled on each edge.}

    We claim that there exists a solution $[n] = P_1 \cup P_2$ to the \Partition{} instance if and only if there exists a \MR{} solution $\rho$ that modifies at most $\abs{S_\rho} \leq 2n$ edges, in which case $\rho$ can be chosen to be increase-only.  To prove this, first consider one of the hexagonal gadget pieces (the reader may find it helpful to refer back to Figure~\ref{figBoundedTreewidthHardnessGadget}). Since each $(x_i + 1) < \frac32$, the triangle inequalities are violated in both of the triangles $a_{i - 1}^1b_i^1c_i^1$ and $a_{i - 1}^2b_i^2c_i^2$, as well as both of the trapezoids $a_{i - 1}^2a_i^1b_i^1b_i^2$ and $a_{i - 1}^2a_i^1c_i^1c_i^2$. Since the $2n$ triangles across all $n$ gadgets are disjoint, any valid $S_\rho$ must contain one edge from each of them. Thus, $\abs{S_\rho} = 2n$ if and only if some choice of a single edge from each triangle satisfies all other constraints. Within the $i\tth$ gadget, the only way to satisfy both trapezoids is to include an edge from the top trapezoid in one triangle and the bottom trapezoid in the other triangle. That is, we must take edges $a_i^1c_i^1$ and $a_{i - 1}^2b_i^2$, or instead take edges $a_i^1b_i^1$ and $a_{i - 1}^2c_i^2$. Note that all of these edges are light edges in their respective broken cycles, so either choice will be an increase-only solution.
    
    For a given choice of $S_\rho$ of this form, let $P_1$ be the set of indices $i$ from which $a_i^1c_i^1$ and $a_{i - 1}^2b_i^2$ are chosen, and $P_2$ be the set of indices $i$ from which $a_i^1b_i^1$ and $a_{i - 1}^2c_i^2$ are chosen. To determine whether this solution is valid, we must examine all other potential cycles $c$ that do not intersect $S_\rho$. For a cycle $c$ that does not cross either of the two edges of weight $4n + \frac12$, it is either contained within one of the hexagons, in which case we have already verified that all constraints are satisfied, or it crosses to a neighboring hexagon and back, involving two edges of weigh 3, which cannot possibly violate the triangle inequality because the only other large edge weight is 6. Thus the only remaining broken cycles to consider are those that cross one of the edges of weight $4n + \frac12$ (but not both). Observe that the endpoints of either of these edges are separated by a minimum of $n$ crossings between hexagons (which have weight 3), and $n$ additional edges (which have weight at least 1). Thus, for such a cycle to violate the triangle inequality, it cannot take any other detours, such as following any of the horizontal or vertical edges in Figure~\ref{figBoundedTreewidthHardnessGadget}. After removing edges in $S_\rho$, there are only two potentially broken cycles remaining. The broken cycle from $a_0^1$ to $a_n^1$ has weight
    $$\sum_{i \in P_1} (3 + (x_i + 1)) + \sum_{i \in P_2} (3 + 1) = 4n + \sum_{i \in P_1} x_i,$$
    while the broken cycle from $a_0^2$ to $a_n^2$ has weight
    $$\sum_{i \in P_1} (1 + 3) + \sum_{i \in P_2} ((x_i + 1) + 3) = 4n + \sum_{i \in P_2} x_i.$$
    Both of these are at most the weight of the heavy edge, $4n + \frac12$, if and only if $P_1 \cup P_2$ is a solution to the \Partition{} instance.
\end{proof}

%% file: planar-is-hard.tex
In this section we focus on the restriction when $G$ is planar. In this setting, we have very strong hardness results:

\begin{theorem}\label{thmPlanarHardness}
	\MR{} and \IOMR{} are both $\NPh$, even when
	\begin{itemize}
		\item The input graph $G$ is planar, with maximum degree 6, and
		\item The weights only take on 3 distinct values, namely 1, 4, and 13.
	\end{itemize}
\end{theorem}
Throughout this section, we denote $|V(G)|, |E(G)|$ by $n, m$ respectively.
We reduce from the following problem, which we show is equivalent to a variant of the maximum independent set problem, \IS. An \emph{$s$-nonsink orientation} of an undirected graph is an assignment of directions to the edges such that at most $s$ vertices have an out-edge (i.e., are not sinks).
\begin{problem}[\pOrientation]
	Given a planar graph $G$ of maximum degree 3, and a positive integer $s$, does $G$ admit an $s$-nonsink orientation?
\end{problem}

\begin{lemma}\label{lemPOrientationHard}
	\pOrientation{} is $\NPh$.
\end{lemma}

\begin{proof}
	Recall that \IS{} is $\NPh$ even on this class of graphs \cite{garey1977rectilinear}. The proof then follows from the following simple observation: A set $S \subseteq V(G)$ is an independent set if and only if there is an orientation of the edges of $G$ such that all vertices in $S$ are sinks. Thus, $G$ has an independent set of size $k$ if and only if it has an edge orientation such that there are at least $k$ sinks, in which case at most $s := n - k$ vertices are \emph{not} sinks.
\end{proof}

\begin{proof}[Proof of Theorem \ref{thmPlanarHardness}]
	Given a instance of \pOrientation{} with graph $G$, we construct a \MR{}/\IOMR{} instance $H$ as follows. For each vertex and edge of $G$, we construct corresponding vertex and edge gadgets in $H$ as illustrated in Figure~\ref{figPlanarHardnessGadget}:
	\ipncm{.8}{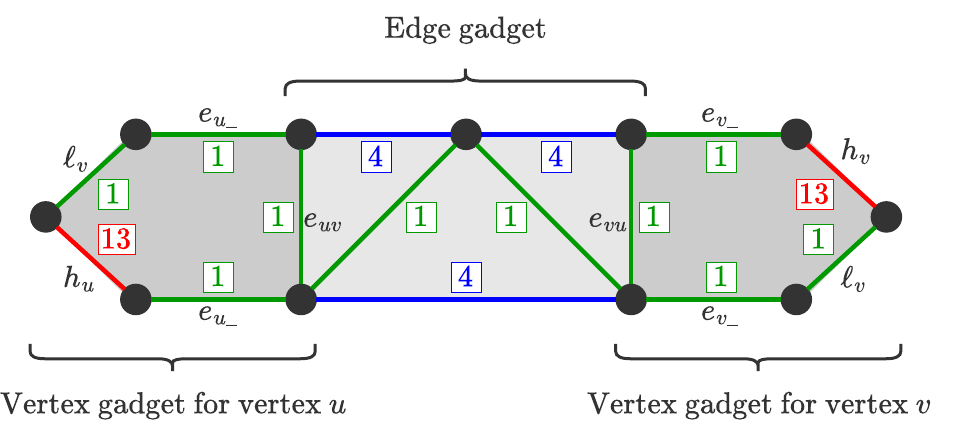}{\label{figPlanarHardnessGadget} The gadgets (i.e., weighted subgraphs of $H$) produced in the reduction in the proof of Theorem \ref{thmPlanarHardness}, for an arbitrary pair of adjacent vertices $u$ and $v$, each of degree 3.}
	\begin{itemize}
		\item A vertex gadget for vertex $u \in V(G)$ consists of a cycle of size 5. Two of the adjacent edges are denoted $\ell_u$ and $h_u$, of respective weights 1 and 13. The other three \emph{terminal} edges are possibly involved in edge gadgets as well. Specifically, for each neighbor $v \in N(u)$, we denote one of the terminals as $e_{uv}$, with weight 1. If there are fewer than three neighbors, there will be some terminals that aren't associated to edges in $G$, which will instead have weight 4.
		\item An edge gadget for edge $uv \in E(G)$ links together terminals $e_{uv}$ and $e_{vu}$ in a chain of 3 triangles, one sharing edge $e_{uv}$, one sharing $e_{vu}$, and a \emph{central} triangle sharing an edge from each of the other two triangles (not $e_{uv}$ or $e_{vu}$). The two terminals and two edges shared between triangles all have weight 1. The other three edges have weight 4.
	\end{itemize}
	An full example of the entire graph $H$ produced by the reduction is shown in Figure~\ref{figPlanarHardnessExample}.
	
	\ipncm{.4}{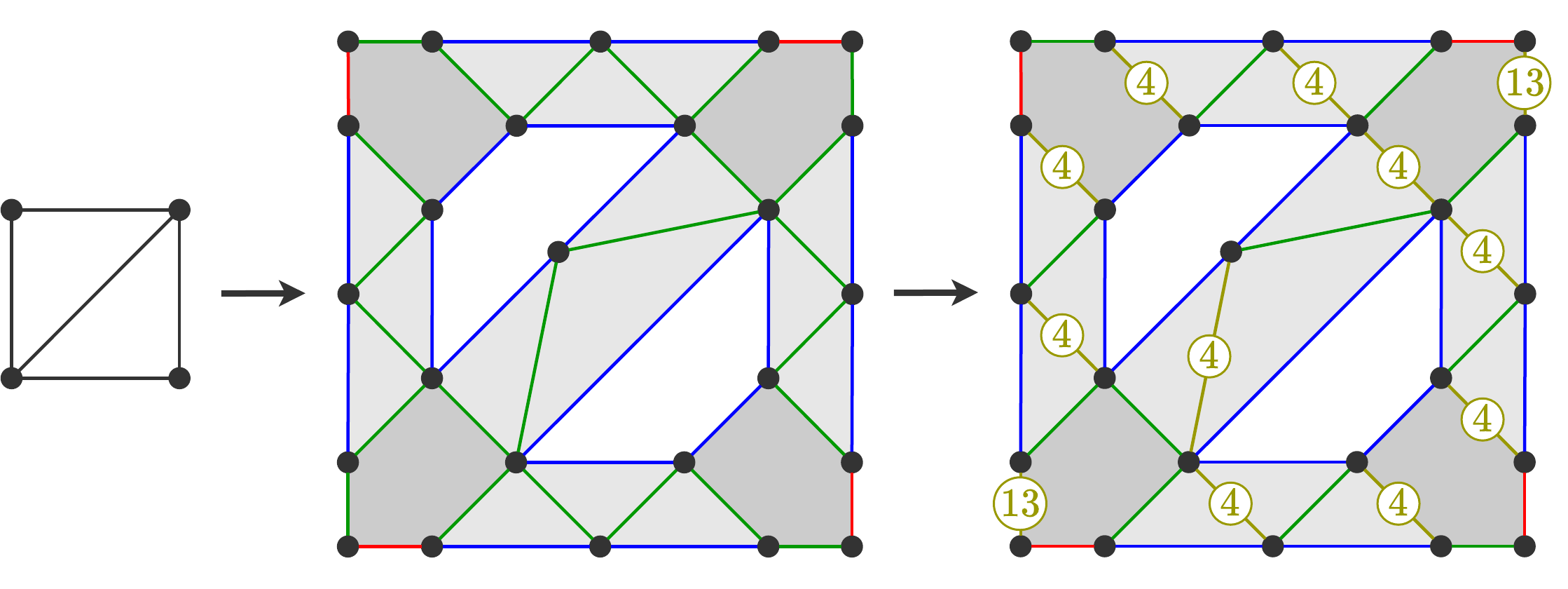}{\label{figPlanarHardnessExample} An example of the weighted graph $H$ (center) produced from a small graph with 4 vertices $G$ (left) according to the reduction from \pOrientation. As in Figure~\ref{figPlanarHardnessGadget}, green edges have weight 1, blue edges have weight 4, and red edges have weight 13. The optimal \MR/\IOMR{} solution (right) modifies $10 + 4 - 2 = 12$ edges, with new distances annotated with the circled numbers. This solution corresponds to the maximal independent set of size 2 in $G$, as constructed in the proof.}
	
	Suppose $G$ has $m$ edges. We will show that $G$ admits an $s$-nonsink orientation if and only if $H$ can be repaired by modifying at most $2m + s$ edges. For the forward direction (also illustrated in Figure~\ref{figPlanarHardnessExample}), fix an $s$-nonsink orientation. We first satisfy all of the broken cycles that are contained entirely within edge gadgets as follows: For an edge that oriented from $u$ to $v$, change the weights of $e_{vu}$ and the common edge of the other two triangles that do not contain $e_{vu}$, setting them both to 4. This turns the edge gadget into 2 disjoint pairs of vertices at distance 1, with all other distances at 4, which is clearly a metric. Next we turn to the vertex gadgets. For each vertex $v$ that is a sink, we have already made each of the three terminals distance 4, so the triangle inequality is satisfied: $13 \leq 1 + 4 + 4 + 4$. For each vertex $v$ that is not a sink, we then increase $\ell_v$, setting it to 13. This clearly makes the vertex gadget cycle into a metric. Thus, after $2m + s$ modifications, we have resolved all broken cycles contained entirely within a vertex or edge gadget. We must finally argue that there are no other potential violations involving multiple gadgets.  First note that there is clearly no path of edges of weight 1 between the endpoints of any edge of weight 4. Thus, all we need to worry about is cycles involving a heavy edge $h_v$ of weight 13  Any cycle that can be contracted to a point via a homotopy through the vertex/edge gadget faces must satisfy the triangle inequality because it can be decomposed as a vector sum of triangles and pentagons within gadgets, which we have just argued all satisfy the triangle inequality themselves. Any cycle that is not contractible crosses at least 3 edge gadgets, so contains three edges of weight 4, and also at least one edge of weight at least 1 (namely the $\ell_v$ next to $h_v$). Hence, the triangle inequality is satisfied.
	
	For the backward direction, we essentially show that this construction is optimal. Formally, let $\rho: E(H) \to \rr_{> 0}$ be a metric repair solution modifying at most $\abs{S_\rho} \leq 2m + s$ edges. Our first step is to edit $\rho$ to obtain a new weighting $\rho': E(H) \to \rr_{> 0}$ with the following properties:
	\begin{enumerate}[(1)]
		\item\label{itmNewSolutionEdges} For each edge $uv \in E(G)$, $S_{\rho'}$ contains exactly two edges in the edge gadget for $uv$. Furthermore, $S_{\rho'} \cap \{e_{uv}, e_{vu}\} \leq 1$ and $\rho'(e_{uv}), \rho'(e_{vu}) \leq 5$. (In other words, one of the terminals is unmodified from $w$ and the other has weight at most 5.)
		\item\label{itmNewSolutionVertices} The triangle inequality is still satisfied within each vertex gadget. (And also within each edge gadget, but we do not need this for the proof.)
		\item\label{itmNewSolutionBetter} The new solution modifies weakly fewer edges: $\abs{S_{\rho'}} \leq \abs{S_{\rho}}$.
	\end{enumerate}
	We construct $\rho'$ iteratively, one edge gadget at a time, ensuring each satisfies Property \ref{itmNewSolutionEdges} while always maintaining Properties (\ref{itmNewSolutionVertices}) and (\ref{itmNewSolutionBetter}). For a given edge gadget corresponding to edge $uv \in E(G)$, there are several cases to consider, depending on how $S_\rho$ intersects the edge gadget:
	
	\begin{itemize}
		\item\textbf{Case 1: There are at most two modified edges in the edge gadget.} In this case, we claim that setting $\rho'(e) = \rho(e)$ for each of the seven edges works; no edits are needed. Observe that, according to the original weights $w$, there are three triangles that must be repaired, and one modification can satisfy at most two of them. Thus, $\rho$ must already have at least two modified edges, so there are exactly two. Furthermore, if one of them is a terminal $e_{uv}$, then the other modified edge must be the edge shared in common with both of the other two triangles. In that case, $\rho(e_{uv}) \leq 5$ by the triangle inequality within the edge gadget triangle, while $e_{vu}$ is unmodified. Thus, Property \ref{itmNewSolutionEdges} is already satisfied, and requires no changes.
	
		\item\textbf{Case 2: There are exactly three modified edges.} Then there are two subcases:
		
		\begin{itemize}
			\item\textbf{Case 2a: At most one of the modified edges is a terminal.} Say that $e_{vu}$ is \emph{not} modified. Then we define $\rho'(e_{uv}) = \rho'(e^*) := 4$, where $e^*$ is the edge shared between the two triangles not containing $e_{uv}$. All the remaining edges remain the same as in $w$. This satisfies Property \ref{itmNewSolutionEdges}, uses one fewer modification, but may violate the constraint in the vertex gadget for $u$. (The constraint for $v$ is still satisfied because $\rho'(e_{vu}) = w(e_{vu}) = \rho(e_{vu})$.) To fix this, we increase either $\ell_u$ or $h_u$ (whichever one has smaller weight under $\rho$) until the triangle inequality is satisfied, preserving Property \ref{itmNewSolutionVertices}. This increases the number of modifications by one, but Property \ref{itmNewSolutionBetter} is still satisfied because we had a surplus of one fewer modification.
			\item\textbf{Case 2b: Both of the terminals are modified.} There is only one other modified edge within the edge gadget, so it must be on the central triangle. Thus it cannot be on both of the other triangles. So suppose it is not on the triangle containing $e_{vu}$. Then we define $\rho'(e_{vu}) = \rho(e_{vu})$ and $\rho'(e^*) := 4$, where $e^*$ is the edge shared between the two triangles not containing $e_{vu}$. All the remaining edges remain the same as in $w$. Observe that this satisfies Property \ref{itmNewSolutionEdges} because none of the edges in the triangle containing $e_{vu}$ have been changed from the valid solution $\rho$, so the triangle inequality implies $\rho'(e_{vu}) = \rho(e_{vu}) \leq 4 + 1 = 5$ As in Case 2a, we can fix the constraint in the vertex gadget for $u$ by using the additional modification we saved, while the constraint for $v$ is still satisfied because $\rho'(e_{vu}) = \rho(e_{vu})$. Thus, properties (\ref{itmNewSolutionVertices}) and (\ref{itmNewSolutionBetter}) are satisfied as before.
		\end{itemize}
		\item\textbf{Case 3: There are four or more modified edges.} Then we can apply the same construction in Case 2a, except that we may need to fix \emph{both} vertex gadgets using the two extra modifications we saved. This obviously satisfies all three properties.
	\end{itemize}
	
	Given $\rho'$, we define an orientation for $G$ that points each edge $uv \in E(G)$ to the vertex whose terminal is modified by $\rho'$, if such an edge exists. The second part of Property \ref{itmNewSolutionEdges} guarantees that this is well defined, as there is at most one such terminal. (If neither have been modified by $\rho'$, the edge can be oriented either way). Property \ref{itmNewSolutionBetter} implies that $\abs{S_{\rho'}} \leq 2m + s$, while the first part of Property \ref{itmNewSolutionEdges} says that exactly $2m$ of these edges are contained within edge gadgets, so at most $s$ vertices of $G$ have vertex gadgets with additional modifications under $\rho'$. We claim that these are the only possible vertices that are not sinks. Suppose toward a contradiction that some vertex $u$ is not a sink, but its vertex gadget does \emph{not} have additional modifications under $\rho'$. There must be some adjacent vertex $v$ with an edge oriented from $u$ to $v$. By the definition of the orientation, this means that $\rho'(t_{uv}) = w(t_{uv}) = 1$. Letting the other two terminal edges be $e_1$ and $e_2$, we know from the third part of Property \ref{itmNewSolutionEdges} that each are weighted at most 5 under $\rho'$ if they are part of an edge gadget, and otherwise weighted 4 by the definition of the reduction, which is also at most 5. Thus, the triangle inequality for the heavy edge $h_u$ is violated:
	$$\rho'(\ell_u) + \rho'(e_{uv}) + \rho'(e_1) + \rho'(e_2) \leq 1 + 1 + 5 + 5 = 12 < 13 = \rho'(h_u)$$
	This contradicts Property \ref{itmNewSolutionVertices}. Hence, we have an $s$-nonsink orientation.
\end{proof}

%% file: grid-is-hard.tex
In this section we show that metric repair is $\NPh$ on grid graphs. We once again use the notation $|V(G)| = n$ and $|E(G)| = m$ where $G$ is the input instance of our reduction. Our starting point is a rectilinear layout of $G$. It is known that every $2$ connected planar graph admits, in linear time, an embedding into an $O(n)\times O(n)$ grid in which every vertex of $G$ is represented as a horizontal path and every edge as a vertical path joining the paths of its endpoints (see~\cite{rosenstiehl1986rectilinear}). We write $H_G$ for this embedding, see \Cref{fig:grid_representation}.

We fix some terminology for grid graphs. We refer to the vertices of a grid graph as \emph{points} and to its edges as \emph{segments}. Points are indexed by $\ZZ\times \ZZ$, and two points are adjacent when exactly one of their coordinates differ by $1$. For a fixed $i$ and integer $\ell$, the points $((i,j+t))_{t=1}^\ell$ induce a \emph{vertical path} of length $\ell$; \emph{horizontal paths} are defined analogously.

This construction alone is not enough: $H_G$ carries far more structure than $G$. The extra grid edges -- lying on no edge or vertex paths -- close cycles with no counterpart in $G$, making it difficult to set weights to the segments of $H_G$ in a way that the broken cycles in $H_G$ are an image of the broken cycles in $G$, with none of the spurious cycles broken. We first reshape $H_G$ to control which cycles can appear, then choose weights.

\begin{figure}[htb]
    \centering
    \begin{subfigure}[c]{0.45\textwidth}
        \centering\includegraphics{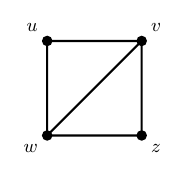}
    \end{subfigure}
    \begin{subfigure}[c]{0.45\textwidth}
        \centering
        \includegraphics{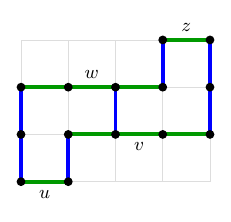}
    \end{subfigure}
    \hfill
    \caption{A planar graph $G$ and one possible rectilinear embedding $H_G$, as in \cite{rosenstiehl1986rectilinear}}.
    \label{fig:grid_representation}
\end{figure}

\subsection{The Construction}
We transform $H_G$ into our final construction $\Gamma$ in two steps, illustrated in \Cref{fig:grid construction}. 

\paragraph{Step 1: Dilation} We first dilate $H_G$ by a factor of $2$, sending each point $(i,j)$ to $(2i-1,2j-1)$ and subdividing every segment: the segment between $(i,j)$ and $(i,j+1)$ becomes the length-$2$ path
\[(2i-1,2j-1),\ (2i,2j-1),\ (2i+1,2j-1),\]
and vertical segments are treated analogously. The purpose of this step is separation: Two edge paths that were adjacent in $H_G$ now lie two columns apart, with a column of new points between them, so no single segment can join to distinct edge paths. If $v \in V(G)$ was represented by a vertex path of length $\ell$, it is now represented by a horizontal path of length $2\ell\eqdef \ell_v$.

\paragraph{Step 2: Blocks.} 
We replace each dilated vertex path by a rectangular subgrid, inflating it vertically and horizontally.For $v\in V(G)$, the \emph{block} $B_v$ is an $(m+1)\times (m+1)\cdot \ell_v$ subgrid: it has $m+1$ rows, and consecutive original points of $v$'s path are now spaced $m+1$ columns apart. We refine the remainder of the grid to match and add the segments needed for the result to be a grid graph. The purpose of this block is robustness: the $(m+1)$ scaling prevents a small solution from using segments from blocks. \Cref{lem:grid-block-segment-hits-atmost-two} and makes this argument precise.\\ 

Note that the resulting graph is of size polynomial in $n,m$, and so the construction can be done in polynomial time.  We denote the resulting grid graph by $\Gamma$, Each $v\in V(G)$ is represented by its \emph{block} $B_v$, and each edge $e\in E(G)$ by a vertical \emph{edge path} $P(e)$. By the dilation of Step 1, distinct edge paths remain separated: if $i_e$ and $i_{e'}$ are the first coordinates of the points in $P(e)$ and $P(e')$, then
\[
|i_e-i_{e'}|\ge 2.
\]

We classify segments of $\Gamma$.
A segment lying inside a block $B_v$ is a \emph{vertex
segment}, and a segment lying in an edge path $P(e)$ is called an
\emph{edge segment}. Every other segment is \emph{auxiliary}. Vertex and edge segment together are \emph{structural}, and a point is structural if it is an endpoint of a structural segment.  Finally, for
$e=uv$ the edge path $P(e)$ meets each of the blocks $B_v$ and $B_u$ in a single point called the \emph{terminal}, $t_{e,v}$ and $t_{e,u}$ respectively. By
construction, every terminal lies on either the top or the bottom
boundary of its block.

The block dimensions are designed to keep the terminals of a block apart. This is the geometric fact that is the crux of the soundness argument (\Cref{lem:grid-block-segment-hits-atmost-two}):
\paragraph{Terminal Separation.} Within any block $B_v$, the first coordinate of two terminals on the same boundary differ by at least $m+1$. The second coordinate of two terminals on opposite boundaries differ by at least $m+1$.

\begin{figure}[htbp]
    \centering
    \begin{subfigure}[c]{0.3\textwidth}
        \centering
        \includegraphics[width=\textwidth]{planar_rectilinear.pdf}
        \label{fig:grid embedding}
    \end{subfigure}
    \hfill
    \begin{subfigure}[c]{0.3\textwidth}
        \centering
        \includegraphics[width=\textwidth]{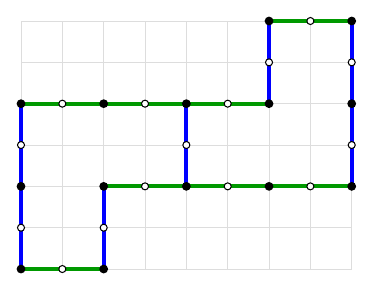}
        \label{fig:factor 2 dilation}
    \end{subfigure}
    \hfill
    \begin{subfigure}[c]{0.3\textwidth}
        \centering
        \includegraphics[width=\textwidth]{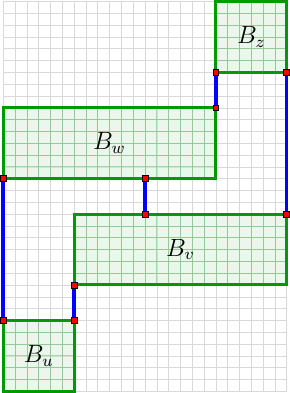}
        \label{fig:block construction}
    \end{subfigure}
    
    \caption{The two steps: On the left is the original rectilinear embedding $H_G$. in the middle is Step $1$, the dilation of $H_G$. On the right is Step $2$, the blocks and refinement of the graph resulting in $\Gamma$. Terminals are depicted in in red (square). The skeleton $\Sigma$ is the graph induced by the colored edges.}
    \label{fig:grid construction}
\end{figure}

\paragraph{Weights.}  Let $N$ be the number of segments in $\Gamma$, and let $\varepsilon = 1/(2N)$. Write $w_{\max} = \max_{e\in E(G)}w(e)$. We define the weight function $w_\Gamma$ on $\Gamma$ as follows.
\begin{itemize}
    \item \textbf{Edge Segments:} On each edge path $P(e)$, fix a distinguished segment $s_e$ and set $w_\Gamma(s_e) = w(e)$. Set the weight of every other segment in $P(e)$ to $\varepsilon$.
    \item \textbf{Vertex Segments:} Set their weights to $\varepsilon$.
    \item \textbf{Auxiliary Segments:} Set their weights to $w_{\max}$. 
\end{itemize}
We denote this weight function by $w_\Gamma$. The only segments of non-negligible weights are the distinguished $s_e$, which hold the original edge weights, and auxiliary edges that are uniformly heavy. The remaining segments weigh $\varepsilon$. In particular, the total $\varepsilon$-mass of $\Gamma$ is at most $N\varepsilon = \frac{1}{2}$.

We define the \emph{Skeleton} of $\Gamma$, denoted $\Sigma$, to be the subgraph on edge and vertex segments. As the weights suggest, broken cycles of $\Gamma$ will turn out to be exactly the images, running along $\Sigma$, of broken cycles of $G$. 
\begin{definition}[Realization]
    For a cycle $C$ in $G$, a \emph{realization} of $C$ in $\Gamma$ is a simple cycle $K$ in $\Gamma$ such that
\begin{enumerate}[label=(\roman*)]
  \item $K$ uses only points of
        $\left(\underset{v\in V(C)}{\bigcup} B_v \right)\cup \left(\underset{e\in E(C)}{\bigcup} P(e)\right)$, and
  \item $K$ traverses $P(e)$ in full for every $e\in E(C)$.
\end{enumerate}
We note that every cycle $C$ has many realizations in $\Gamma$.
\end{definition}

The proof is organized as follows. We first formalize the above claim and show that the broken cycles of $\Gamma$ are precisely the realizations of broken cycles of $G$ (\Cref{lem:grid-broken-cycles}). Completeness of the reduction follows immediately. 

For soundness, we show multiple exchange arguments that allow us to change an arbitrary metric repair solution in $\Gamma$, to one that does not include auxiliary segments (\Cref{lem:grid-mr-solution-structure}) and that does not include vertex segments (\Cref{lem:grid-can-remove-blocks}). The main technical lemma is \Cref{lem:grid-block-segment-hits-atmost-two}, which serves as a building block towards removing vertex segments. We conclude the soundness proof in \Cref{lem:grid-soundness}, which together with the completeness argument (\Cref{lem:grid-completeness}) concludes the reduction.
\begin{claim}[Induced Skeleton]\label{claim: grid-induced-skel}
The skeleton $\Sigma$ is an induced subgraph of $\Gamma$; equivalently, no auxiliary
segment has both of its endpoints structural.
\end{claim}
\begin{proof}
Let $a=pq$ be an auxiliary segment and suppose $p,q$ are both structural. Each block is a subgrid of $\Gamma$, so any segment between two points of the same block is itself a block
segment, hence structural; thus $p,q$ do not lie in a common block. Each edge path is an induced path, so any segment between two of its points is one of its own segments; thus
$p,q$ do not lie on a common edge path. A block $B_v$ and an edge path $P(e)$ with $v\in e$ meet only at the terminal $t_{e,v}$, and a segment incident to a terminal is structural; while a block and a non-incident edge path, and two distinct edge paths, are kept at distance $\ge 2$ by the dilation. Hence no single segment joins them. In every case $a$ would have to be structural, a contradiction.
\end{proof}

\begin{lemma}\label{lem:grid-broken-cycles}
    A cycle in $\Gamma$ is broken if and only if it is a realization of a broken cycle in $G$
\end{lemma}
\begin{proof}
Let $K$ be a realization of a cycle $C$ in $G$. By definition $K$ uses no auxiliary
segment and traverses $P(e)$ in full for each $e\in C$. Hence its segments are the
distinguished segment $s_e$ of weight $w(e)$ for each $e\in C$, together with a set $Z$ of $\varepsilon$-weight segments---the remaining edge-path segments and the block
segments joining consecutive edge paths. Recall that $N$ is the number of segments in $\Gamma$. Since $|Z|\le N-1 < N$, and all structural edges weigh $\varepsilon = 1/(2N)$ we have $ 0\le |Z|\varepsilon < 1/2$. Let $Y \defeq |Z|\varepsilon$, and let $e_{\max}$ be a heaviest edge of $C$. 
Clearly $s_{e_{\max}}$
is a heaviest segment of $K$. Let $\deficit(C) = w(e_{\max}) - \sum_{e\in C\setminus e_{\max}}$
and $\deficit(K) = w(s_{e_{\max}}) - \sum_{s\in K\setminus s_{e_{\max}}} w_\Gamma(s)$. Then
\[
  \deficit(K) =  w(e_{\max})-\sum_{e\in C\setminus e_{\max}}w(e) - Y
   = \deficit(C)-Y .
\]
Because $G$ has integer weights, $\deficit(C)\in\mathbb{Z}$. Therefore if $C$ is broken
then $\deficit(C)\ge 1$, so $\deficit(K)>\tfrac12>0$ and $K$ is broken. If $C$ is not
broken then $\deficit(C)\le 0$, so $\deficit(K)\le -Y\le 0$ and $K$ is not broken.
Conversely, we show that every broken cycles in $\Gamma$ is a realization of a broken cycle in $G$. Let $K$ be a broken cycle in $\Gamma$. First, $K$ cannot have more than one auxiliary segment: Since auxiliary segments are of maximum weight in $\Gamma$, a cycle that contains two cannot be broken. Moreover, if $a = pq$ is the unique auxiliary segment in $K$, the segments of $K$ incident to $p,q$ other than $a$ are structural. By Claim \ref{claim: grid-induced-skel}, $\Sigma$ is an induced graph of $\Gamma$, so $a$ must be structural -- a contradiction. Therefore, $K$ remains entirely in $\Sigma$. If $K$ remains within a single block, all of its weights are identical and clearly cannot be broken. Therefore $K$ leaves every block it meets. Each interior point of an edge path $P(e)$ has degree $2$ in $\Sigma$, since its only skeleton neighbors are its two neighbors along $P(e)$, so once $K$ enters $P(e)$ at a terminal it must traverse $P(e)$ to the opposite terminal. Thus $K$ traverses in full every edge path it meets. Reading these edge paths in the cyclic order in which $K$ meets them yields a closed walk in $G$; since $K$ is simple and distinct edge paths are segment-disjoint, each $P(e)$ is used at most once, so the walk is a simple cycle $C$ in $G$ with $E(C)=\{e: P(e)\subseteq K\}$. By construction $K$ uses only points of $\bigcup_{v\in V(C)}B_v\cup\bigcup_{e\in E(C)}P(e)$ and traverses each such $P(e)$, so $K$ is a realization of $C$. 
\end{proof}
\subsection{Completeness}
\begin{lemma}[Completeness]\label{lem:grid-completeness}
    Let $M\subset E$ be a metric repair solution in $G$. Then $M_\Gamma \defeq \cbk{s_e\mid e\in M}$ is a metric repair solution in $\Gamma$.
\end{lemma}
\begin{proof}
    Since $M$ is a metric repair solution in $G$, it intersects every broken cycle in $G$. Since every realization of a broken cycle traverses its edge paths, $M_\Gamma$ intersects all realizations of all broken cycles. By \Cref{lem:grid-broken-cycles} every broken cycle in $\Gamma$ realizes some broken cycle in $G$, so $M_\Gamma$ hits them all. Clearly $|M| = |M_\Gamma|$, concluding the proof.
\end{proof}
\subsection{Soundness}
Let $S$ be a metric repair solution in $\Gamma$. We need to show that there exists a metric repair solution of at most the same size in $G$. We start by stating an observation used repeatedly throughout multiple arguments of the proof.

\begin{observation}\label{obs:block-separation-argument}
Let $K$ be a broken cycle in $\Gamma$ that realizes a cycle $C$ in $G$, and let $B=B_v$ be
a block met by $K$. Since $K$ is simple, $K\cap B$ is a single arc, joining the two
terminals $p,q$ through which $K$ enters and leaves $B$. For any simple $p$--$q$ path
$\pi\subseteq B$, the cycle $K'$ obtained from $K$ by replacing the arc $K\cap B$ with
$\pi$ is again a realization of $C$, and hence is broken. Moreover, if $\pi$ and the arc
$K\setminus B$ both avoid a set of segments $S$, then $K'\cap S=\varnothing$.
\end{observation}

\begin{proof}
Because $K$ realizes $C$ and meets $B_v$, we have $v\in V(C)$. The arc $K\setminus B$
meets $B$ only at $p,q$, and $\pi\subseteq B$, so $K'$ is a simple cycle. It traverses
every edge path $P(e)$, $e\in E(C)$, in full---these lie on $K\setminus B$, which is
unchanged---and uses only points of $B_v$ and of $K\setminus B$, all admissible for $C$;
hence $K'$ realizes $C$. By \Cref{lem:grid-broken-cycles}, $K'$ is broken. The final clause is immediate, as
$K'\subseteq\pi\cup(K\setminus B)$.
\end{proof}

\begin{lemma}[Minimal Properties of $S$]\label{lem:grid-mr-solution-structure}
Let $S$ be a metric repair solution in $\Gamma$. Then there exists a metric repair solution $X$ in $\Gamma$ with the following properties:
\begin{enumerate}
    \item $|X|\le |S|$
    \item $X$ contains no auxiliary segments.
    \item for each $e\in E(G)$, $|X\cap P(e)| \le 1$.
\end{enumerate}
\end{lemma}
\begin{proof}
    Let $S$ be a metric repair solution in $\Gamma$. Recall that $S$ is a metric repair solution iff it is a hitting set for the collection of broken cycles in $\Gamma$. By \Cref{lem:grid-broken-cycles}, no auxiliary segment participates in a broken cycle, hence $X = S\cap \Sigma$ is still a solution. We iteratively make $X$ smaller, and maintain a solution. For each $e\in E$, if $|P(e)\cap X|\ge 2$, then set $X \gets (X\setminus P(e)) \cup \cbk{s_e}$. We claim that after each such change, $X$ is still a metric repair solution for $\Gamma$. Every broken cycle $K$ in $\Gamma$ that intersects $X$ but does not intersect $X\setminus P(e)$ must contain the path $P(e)$, and in particular $s_e$. Then $(X\setminus P(e))\cup \{s_e\}$ intersects $K$ precisely in $s_e$.
\end{proof}
It is left to reduce the solution to include no vertex segments. We do so in the subsequent lemmas. 

\begin{lemma}\label{lem:grid-block-segment-hits-atmost-two}
    Let $S$ be a metric repair solution that contains only structural segments, at most one per edge path. Let $B=B_u$ be a block
    with $B\cap S\ne \emptyset$. Suppose there is a segment $r\in S\cap B$ such that:
    \begin{enumerate}
        \item $r$ hits two broken cycles $K_a$ and $K_b$;
        \item $S\cap K_a=S\cap K_b=\{r\}$;
        \item $K_a,K_b$ are realizations of broken cycles
        $C_a,C_b$ in $G$, and the edges of $C_a$ and $C_b$
        incident to $u$ are all distinct.
    \end{enumerate}
    Then $|S|>m$.
\end{lemma}

\begin{proof}
    Suppose by way of contradiction that $|S|\le m$. Let $a_1,a_2$ be
    the terminals through which $K_a$ enters and leaves $B$, and let
    $b_1,b_2$ be the corresponding terminals for $K_b$. Since the edges of $C_a$ and $C_b$ incident to $u$ are all distinct, the four terminals $a_1,a_2,b_1,b_2$ are distinct.

    We first show that $S\cap B$ separates $a_1$ from $a_2$ inside $B$, and likewise $b_1$ from $b_2$.
    Since $S\cap K_a=\{r\}\subseteq B$, the arc $K_a\setminus B$ avoids $S$. If some path $\pi\subseteq B\setminus S$ joined $a_1$ to $a_2$, then by \Cref{obs:block-separation-argument}, replacing the arc $K_a\cap B$ by $\pi$ would produce a broken cycle $K_a'$ realizing $C_a$ with $K_a'\cap S=\varnothing$, contradicting that $S$ is a metric repair solution, therefore $S\cap B$ separates $a_1$ from $a_2$. the argument for $b_1,b_2$ is identical.



    Write $r=xy$, and let $Q_1,Q_2$ be the connected components of $B\setminus S$ containing $x,y$ respectively. Traversing $K_a$ from $a_1$ to $a_2$, label $x,y$ so that $x$ is reached first. The subarc from $x$ to $a_1$ avoids $S$ and therefore $a_1\in Q_1$. Similarly the sub-arc from $y$ to $a_2$ avoids $S$, so $a_2\in Q_2$. Since $S\cap B$ separates $a_1,a_2$, $Q_1\neq Q_2$. Relabeling $b_1,b_2$ if necessary so that $x$ is reached first along $K_b$, we have $b_1\in Q_1$ and $b_2\in Q_2$.

The $m+1$ rows of $B$ are pairwise edge-disjoint horizontal paths and its $(m+1)\ell_u$ columns are pairwise edge-disjoint vertical paths. By cut--cycle duality, a cut separating $\{a_1,b_1\}$ from $\{a_2,b_2\}$ corresponds to a cycle in the dual $B^*$. As established, $S\cap B$ is such a cut, and let its cycle dual be $\gamma$. Recall that each dual segment of $\gamma$ crosses a primal segment of $B$. Since all four terminals lie on the boundary of $B$, $\gamma$ must pass through the outer face of $B^*$. Let $s,t$ be the primal segments $\gamma$ crosses when it leaves and enters the outer face. These dual segments split the boundary of $B$ into two arcs, $\alpha, \overline{\alpha}$,  containing $\{a_1,b_1\}$ and $\{a_2,b_2\}$ respectively, see \Cref{fig:GridhardLemma}.



It remains to see that the separation property forces $\gamma$ to cross at least $m+1$ edges.
Suppose first that $a_1,b_1$ lie on the same boundary side of $B$. Then $\alpha$ runs along this side. $s,t$ touch $\alpha$ at its endpoints, so $s,t$ flank $a_1,b_1$.  By the construction of $B$, $a_1,b_1$ are at least least $m+1$ columns apart. Each dual segment of $\gamma$ crosses at most one column edge, so $\gamma$ spans at least $m+1$ segments in order to traverse from a columns of $s$ to a column of $t$. Thus $|\gamma| \ge m+1$.

Suppose now $a_1,b_1$ lie on opposite boundaries. Then $\alpha$ traverses one vertical side of the bounadry, hence $\gamma$ separates the left and right vertical sides of the boundary. Therefore the dual segments of $\gamma$ must cross each of the $m+1$ horizontal rows of $B$, and consequently $|\gamma| \ge m+1$.
In either case $\gamma$ crosses at least $m+1$ segments, so $|S|\ge|S\cap B|= |\gamma|\ge m+1$, contradicting $|S|\le m$. 
\end{proof}
\begin{figure}[ht]
    \centering
    \includegraphics[scale=0.4]{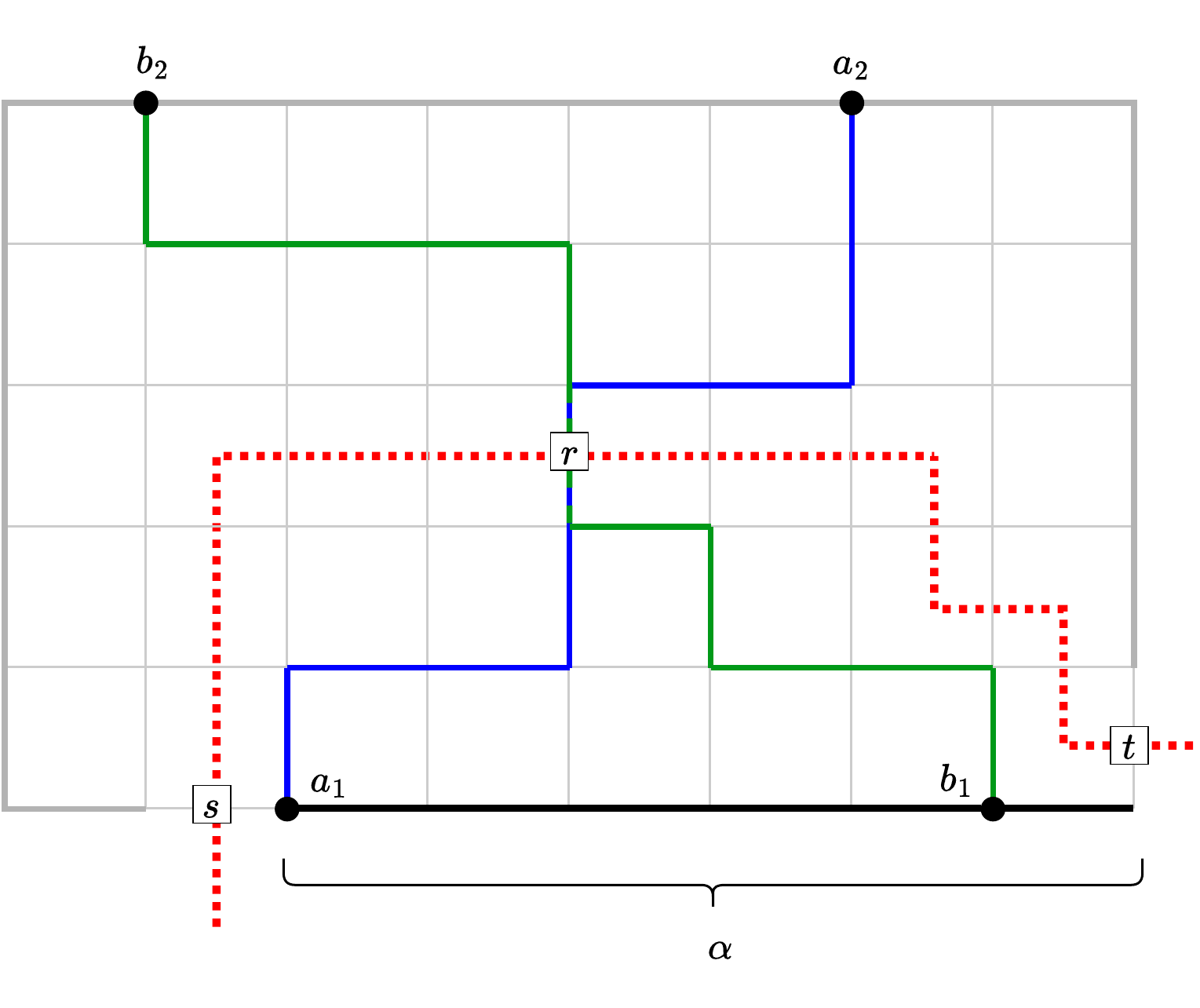}
    \caption{Illustration of \Cref{lem:grid-block-segment-hits-atmost-two}. The paths between terminals are depicted in blue and green, the shared segment is labeled $r$. The segments through which the dual path (dashed red) enters and leaves $B$ are labeled $s,t$, and the arc they define containing $a_1,b_1$ is in black, labeled $\alpha$. $\overline{\alpha}$ is the remainder of the boundary, in bold gray.}
    \label{fig:GridhardLemma}
\end{figure}

We proceed to remove the block segments from $S$. Note that a key property is that if a broken cycle enters the block $B_u$ through the terminal $t_{e,u}$, then it contains the entire edge path $P(e)$.
\begin{lemma}\label{lem:grid-can-remove-blocks}
    Suppose $|S|\le m$.
    Let $r\in S$ be a vertex segment contained in a block $B=B_u$.
    Suppose $\Gamma\setminus (S\setminus \{r\})$ contains broken
    cycles. Then all such broken cycles share at least one edge path
    $P(e)$.
\end{lemma}

\begin{proof}
    Let $\mathcal K$ be the set of broken cycles in
    $\Gamma\setminus (S\setminus \{r\})$. Since $S$ is a metric
    repair solution, every cycle in $\mathcal K$ must contain the
    segment $r$; otherwise it would also be a broken cycle in
    $\Gamma\setminus S$. Moreover, no cycle in $\mathcal K$ contains
    any segment of $S\setminus\{r\}$.

    If $|\mathcal K|=1$, the unique broken cycle $K$ contains an edge path $P(e)$ and the claim is immediate. Assume therefore
    that $|\mathcal K|\ge 2$. Suppose two cycles $K,K'\in\mathcal K$ used four distinct terminals. By \Cref{lem:grid-broken-cycles}, $K$ realizes a broken cycle $C_K$ in $G$, and therefore enters and leaves $B_u$ at the two terminals of the edges of $C_K$ incident to $u$ (similarly $K'$ realizes $C_{K'}$). Then $S\cap K=S\cap K'=\{r\}$ and the edges of $C_K,C_{K'}$ incident to $u$ are all distinct, so $r$ satisfies the hypotheses of \Cref{lem:grid-block-segment-hits-atmost-two} and $|S|>m$, contradicting $|S|\le m$. Hence the terminal pairs used by $\mathcal K$ are pairwise intersecting.

    We claim that all terminal pairs used by cycles in $\mathcal K$
    share a common terminal. Suppose not. Since the terminal pairs are
    two-element sets and are pairwise intersecting, there must be three
    terminals $t_1,t_2,t_3$ and three cycles
    $K_{1,2},K_{1,3},K_{2,3}\in\Kk$, where $K_{i,j}$ enters
    and leaves $B$ through $t_i$ and $t_j$.

    For each pair $i,j\in\{1,2,3\}$, let $Q_{i,j}$ be the subpath of
    $K_{i,j}\cap B$ connecting $t_i$ to $t_j$. Each path
    $Q_{i,j}$ contains the segment $r$, and no segment of
    $Q_{i,j}\setminus\{r\}$ belongs to $S$. Consider the subgraph
    \[
        H \defeq Q_{1,2}\cup Q_{1,3}\cup Q_{2,3}.
    \]
    Since all three paths contain the segment $r$, the graph
    $H\setminus\{r\}$ has at most two connected components. Hence two
    of the terminals $t_1,t_2,t_3$ lie in the same connected component
    of $H\setminus\{r\}$. Without loss of generality, suppose these
    terminals are $t_1$ and $t_2$. Since $H\cap S = \{r\}$, there exists some $t_1-t_2$ path in $H\setminus\{r\}$ that avoids $S$ entirely. Let $\pi$ be one such path.

    Now apply \Cref{obs:block-separation-argument} to $K_{1,2}$, which realizes the broken cycle
$C_{1,2}$: its in-block arc is $Q_{1,2}$ and its complementary arc $K_{1,2}\setminus B$
avoids $S$ . Replacing its in-block arc $Q_{1,2}$ by $\pi$
yields a cycle that realizes $C_{1,2}$, hence is broken, and avoids $S$. This is a broken
cycle in $\Gamma\setminus S$, contradicting that $S$ is a metric repair solution.

Therefore all terminal pairs share a common terminal $t$. As $t=t_{e,u}=P(e)\cap B$ for
the edge $e$ incident to $u$ with $P(e)\cap B=t$, every $K\in\mathcal K$ enters or leaves
$B$ through $t$ and hence contains the whole edge path $P(e)$. Thus all broken cycles in
$\Gamma\setminus(S\setminus\{r\})$ share $P(e)$.

\end{proof}

\begin{corollary}\label{cor:grid-solution-only-edges}
    If $\Gamma$ has a metric repair solution of size at most $m$, then there exists a metric repair solution of size at most $m$ that only uses edge segments. 
\end{corollary}

\begin{proof}
    Let $S$ be a metric repair solution for $\Gamma$ of size at most
    $m$. By \Cref{lem:grid-mr-solution-structure}, we may assume that
    $S$ contains no auxiliary segments. Thus every segment of $S$ is
    either a vertex segment or an edge segment.

    We show that if $S$ contains a vertex segment, then we can replace
    one such segment by an edge segment and still obtain a metric repair
    solution of size at most $m$. Repeating this operation eventually
    produces a solution consisting only of edge segments.

    Let $r\in S$ be a vertex segment, and let $B$ be the block
    containing $r$. Consider
    \[
        S^- := S\setminus\{r\}.
    \]
    If $\Gamma\setminus S^-$ contains no broken cycle, then $S^-$ is
    already a metric repair solution. In this case we may simply remove
    $r$, obtaining a solution of smaller size.

    Otherwise, $\Gamma\setminus S^-$ contains broken cycles, and every such broken cycle must contain the segment $r$. By
    \Cref{lem:grid-can-remove-blocks}, all these broken cycles share an
    edge path $P(e)$. Choose any edge segment $s$ of $P(e)$, and set
    \[
        S' := S^-\cup\{s\}.
    \]
    Then $|S'|\le |S|\le m$. Moreover, every broken cycle in
    $\Gamma\setminus S^-$ contains the edge path $P(e)$, and hence
    contains the segment $s$. Therefore no broken cycle remains in
    $\Gamma\setminus S'$, so $S'$ is a metric repair solution.
    Repeating this operation for all vertex segments in $S$ concludes the proof.
\end{proof}

\begin{lemma}[Soundness]\label{lem:grid-soundness}
    Let $S$ be a metric repair solution on $\Gamma$. Then there exists
    a metric repair solution $M$ on $G$ with $|M|\le |S|$.
\end{lemma}

\begin{proof}
    Suppose $|S| > m$. Any metric repair solution in $G$ has size at most $|E(G)| = m  < |S| $, so the claim holds trivially. Suppose then that $|S|\le m$.
    By the replacement argument of \Cref{cor:grid-solution-only-edges},
    we may assume that $S$ consists only
    of edge segments. Define
    \[
        M \defeq \{e\in E(G) : P(e)\cap S\neq \emptyset\}.
    \]
    Since distinct edge paths $P(e)$ are segment-disjoint, every edge
    $e\in M$ accounts for at least one distinct segment of $S$. Hence
    \[
        |M|\le |S|.
    \]

    We claim that $M$ is a metric repair solution on $G$. Let $C$ be
    a broken cycle in $G$, and let $K$ be a realization of $C$ in
    $\Gamma$. By the construction of realizations, $K$ is a broken
    cycle in $\Gamma$. Since $S$ is a metric repair solution on
    $\Gamma$, there exists some $s\in K\cap S$. Since $S$ consists only of edge segments, $s\in P(e)$ for some edge $e$ and therefore $e\in M$. Moreover, because
    $P(e)$ is contained in the realization $K$, the edge $e$
    belongs to the original cycle $C$. Therefore
    \[
        M\cap C\neq \emptyset.
    \]
    Since every broken cycle $C$ in $G$ is hit by $M$, $M$ is a
    metric repair solution on $G$.
\end{proof}

%% file: conclusion.tex
\section{Conclusion and Open Problems}
\label{sec:conclusion}

Our results locate the tractability boundary of metric repair along two structural axes: planarity offers no relief, while bounded treewidth helps only in
combination with bounded (unary) weights. Several natural questions
remain open on both sides of this boundary.

The most immediate gap concerns small treewidth with unrestricted
weights. Our algorithm for series-parallel graphs runs in time
$O(W^4 m)$ and is therefore only pseudopolynomial, while our hardness
construction in Section~\ref{sec:bdd-tw-hard} shows weak NP-hardness
already at pathwidth $6$. This leaves open whether
\IOMR{} admits a genuinely polynomial-time
algorithm, one whose running time is polynomial in the binary encoding
of the weights, on graphs of treewidth $2$, or even on outerplanar
graphs. A positive answer would require replacing our demand profiles,
whose number is inherently polynomial in $W$, with a representation
whose size is polynomial in the input encoding; a negative answer would
be a striking instance of a problem that is weakly NP-hard even at
treewidth $2$. In the same vein, we do not know where between treewidth
$2$ and pathwidth $6$ the weak hardness truly begins: our
\Partition{} gadgets require width $6$, and it is plausible that a
tighter construction, or conversely a stronger algorithmic technique,
could close this gap.

On the hardness side, our planar reduction uses only three distinct
edge weights ($1$, $4$, and $13$), and it is natural to ask whether
this can be pushed to two. Bi-valued metric repair is an appealing
minimal setting: with a single weight value the problem is trivial, so
two values is the first case in which broken cycles can exist at all,
and the structure of broken cycles reduces to the exact same problem as \textsc{Length-Bounded Multicut}. Resolving the
complexity of bi-valued \MR{} on planar graphs would sharpen our
understanding of exactly how much weight structure is needed to encode hardness. 

More broadly, our results invite a systematic look at problems whose
defining combinatorial objects are determined jointly by the weights
and the topology of the input graph. Metric repair is a hitting-set
problem, but the family being hit, the broken cycles, is not a
structural feature of the graph alone: which cycles are broken, and
which of their edges are light, depends entirely on the interaction
between the weight function and the cycle space. This places the
problem outside the reach of the standard structural meta-theorems.
Courcelle's theorem and its relatives handle properties definable from
the graph structure (possibly with weights entering only through an
objective to be optimized), while purely numeric problems such as
\Partition{} are governed by the arithmetic of the weights with
essentially no graph structure at all. Metric repair sits genuinely in
between and our results quantify this.
